\documentclass[aps,pra,reprint, superscriptaddress,nofootinbib]{revtex4-1}

\usepackage{amsthm}
\usepackage{amssymb}
\usepackage{amsmath}
\usepackage{graphicx}
\usepackage{enumitem}
\usepackage{bbm}
\usepackage{braket}
\usepackage{color}

\newtheorem{lemma}{Lemma}

\newtheorem{proposition}{Proposition}
\newtheorem{defn}{Definition}
\newtheorem{prop}{Proposition}\def\PRO{\begin{prop}}\def\ORP{\end{prop}}
\newtheorem{coro}{Corollary}\def\COR{\begin{coro}}\def\ROC{\end{coro}}
\newtheorem{theo}{Theorem}\def\TH{\begin{theo}}\def\HT{\end{theo}}
\def\TH{\begin{theo}}\def\HT{\end{theo}}
\newtheorem{defi}[prop]{Definition}\def\DE{\begin{defi}}\def\ED{\end{defi}}

\newtheorem{lemme}[prop]{Lemma}\def\LE{\begin{lemme}}\def\EL{\end{lemme}}

\def\ket#1{\left| #1 \right\rangle}
\def\density#1{\left| #1 \right\rangle\left\langle #1 \right|}
\def\bra#1{\left\langle #1 \right|}

\newcommand{\beq}{\begin{equation}}
\newcommand{\Tr}{\text{Tr}}
\newcommand{\eeq}{\end{equation}}

\begin{document}
	\title{Imperfect 1-out-of-2 quantum oblivious transfer: bounds, a protocol, and its experimental implementation
	}
	\author{Ryan Amiri}
	\affiliation{SUPA, Institute of Photonics and Quantum Sciences, Heriot-Watt University, Edinburgh EH14 4AS, United Kingdom}
	\author{Robert St\'{a}rek}
	\affiliation{Department of Optics, Palacky University, 17. listopadu 1192/12, 77900 Olomouc, Czech Republic}
	\author{David Reichmuth}
	\affiliation{SUPA, Institute of Photonics and Quantum Sciences, Heriot-Watt University, Edinburgh EH14 4AS, United Kingdom}
    \author{Ittoop V. Puthoor}	
    \affiliation{SUPA, Institute of Photonics and Quantum Sciences, Heriot-Watt University, Edinburgh EH14 4AS, United Kingdom}
	\author{Michal Mi\v{c}uda}
	\affiliation{Department of Optics, Palacky University, 17. listopadu 1192/12, 77900 Olomouc, Czech Republic}
		\author{Ladislav Mi\v{s}ta,~Jr.}
	\affiliation{Department of Optics, Palacky University, 17. listopadu 1192/12, 77900 Olomouc, Czech Republic}
	\author{Miloslav Du\v{s}ek}
	\affiliation{Department of Optics, Palacky University, 17. listopadu 1192/12, 77900 Olomouc, Czech Republic}
	\author{Petros Wallden}
	\affiliation{LFCS, School of Informatics, University of Edinburgh, 10 Crichton Street, Edinburgh EH8 9AB, United Kingdom}
	\author{Erika Andersson}
	\affiliation{SUPA, Institute of Photonics and Quantum Sciences, Heriot-Watt University, Edinburgh EH14 4AS, United Kingdom}
	\email{ra2@hw.ac.uk}

	\begin{abstract}
		Oblivious transfer is an important primitive in modern cryptography. Applications include secure multiparty computation, oblivious sampling, e-voting, and signatures. 
		Information-theoretically secure perfect 1-out-of 2 oblivious transfer is impossible to achieve. Imperfect variants, where both participants' ability to cheat is still limited, are possible using quantum means while remaining classically impossible. Precisely what security parameters 
		are attainable remains unknown. 
		We introduce a  
		theoretical framework for studying semi-random quantum oblivious transfer, which is shown equivalent to regular oblivious transfer in terms of cheating probabilities. We then use it to derive bounds on cheating. We also present a protocol with lower cheating probabilities than previous schemes, together with its optical realisation. 
		We show that a lower bound of 2/3 on the minimum achievable cheating probability can be directly derived for semi-random protocols using  a 
		different method and definition of cheating than used previously. The lower bound increases from $2/3$ to approximately $0.749$ if the states output by the protocol are pure and symmetric. 
		The oblivious transfer scheme we present uses unambiguous state elimination measurements and can be implemented with the same technological requirements as standard quantum cryptography. In particular, it does not require honest participants to prepare or measure entangled states. 
		The cheating probabilities  are 3/4 and approximately 0.729 for sender and receiver respectively, which is lower than in existing protocols. Using a photonic test-bed, we have implemented the protocol with honest parties, as well as optimal cheating strategies. 
		Due to the asymmetry of the receiver's and sender's cheating probabilities, the protocol can be combined with a ``trivial" protocol to achieve an overall protocol with lower average cheating probabilities of approximately 0.74 for both sender and receiver. This demonstrates that interestingly, protocols where the final output states are pure and symmetric are not optimal in terms of average cheating probability.
	\end{abstract}
	
	\maketitle
	
	\section{Introduction}

	Following the discovery of quantum key distribution in 1984 \cite{bb84}, there arose a general optimism that quantum mechanics may provide a means to perform multiparty computations with information-theoretic security. Despite this early confidence, the history of secure two-party computations is characterised by mainly negative results. Mayers and Lo \cite{May97,Lo97} proved that all one-sided two-party computations are insecure in the quantum setting, meaning that it is impossible to perform important protocols such as bit commitment and oblivious transfer (OT) with information-theoretic security. Nevertheless, imperfect variants of these protocols remain possible, and it has been an interesting and productive open question to determine the optimal security parameters achievable for some important two-party computations. 
	
	For many cryptographic primitives, this question has been definitively answered. For strong coin flipping, Kitaev \cite{Kitaev} introduced the semi-definite programming formalism to show that the product of Alice's and Bob's cheating probabilities must be greater than $1/2$, implying that the minimum cheating probability is at least $1/\sqrt{2}$. For weak coin flipping, Mochon \cite{Moc07} showed that the minimum cheating probability is at least $1/2 + \epsilon$ for any $\epsilon>0$. In the same paper a protocol achieving this bound is presented, showing that the bound is tight. Chailloux and Kerenidis \cite{Cha09} used these results on weak coin flipping to generate a protocol for strong coin flipping achieving Kitaev's bound. Lastly, for quantum bit commitment, Chailloux and Kerenidis \cite{Cha11} proved that the minimum cheating probability is $0.739$, and presented a protocol achieving this bias. Thus, for bit commitment, weak coin flipping, and strong coin flipping the achievability bounds are tight with the known protocols.
	
	For OT on the other hand, the situation is not as clear. Classically, it is impossible to achieve even limited security for OT in the information-theoretic setting, since one party can always cheat with certainty. On the other hand, quantum mechanics allows for imperfect protocols, in which the participants are able to cheat but their abilities are limited. 
	
	OT is a fundamental primitive in cryptography. Its importance stems from the fact that it can be used as the foundation for secure two-party computations; with oblivious transfer, all secure two-party computations are possible \cite{GV88,Kil88}. 
	OT exists in many different flavours, all with slightly different definitions and notions of security. 
	It was first introduced informally in 1970 by Wiesner as ``a means for transmitting two messages either but not both of which may be received" \cite{Wie83}, and subsequently formalised as 1-out-of-2 oblivious transfer (1-2 OT) in \cite{EGL85}. In related work, Rabin \cite{Rab81} introduced a protocol (now called Rabin OT), which was later shown by Cr{\'e}peau \cite{Cre88} to be classically equivalent to 1-2 OT, in the sense that if it is possible to do one, it is possible to use this to implement the other. Various ``weaker" variants of OT have also been proposed, most notably Generalised OT, XOR OT and Universal OT \cite{BC97}, but all have been shown to be equivalent to 1-2 OT \cite{BCW03} in the classical setting. The equivalence is believed to also hold in the quantum setting, but the reduction proofs may need to be revised. There is also work by Damg{\aa}rd, Fehr, Salvail and Schaffner \cite{DFSS05} who define OT in a slightly different way, and characterise security in terms of information leakage. With these definitions (and their quantum counterparts), the authors describe a 1-2 OT protocol which is secure in the bounded quantum storage model. 
	Spacetime-constrained quantum OT protocols have also been proposed~\cite{Garcia}, requiring agents at different locations in spacetime, giving constraints on where in spacetime bit values can be obtained. 
	Recently, a device-independent quantum XOR oblivious transfer protocol was proposed \cite{KST20}. The protocol uses a shared entangled state to reveal cheating. Another version of ``imperfect'' oblivious transfer was considered and experimentally implemented in \cite{RKB+18}, where the authors could achieve vanishing cheating advantage for both sides, at the expense of having a protocol that sometimes fails during honest execution.
	
	In this paper we consider stand-alone quantum protocols for 1-2 OT, including an experimental implementation of such a protocol, and are concerned only with information-theoretic security. As mentioned above, perfect security in this setting is impossible. The best known lower bound on the achievable bias in 1-2 OT protocols is due to Chailloux, Gutoski and Sikora \cite{CGS13}, who show that the minimum cheating probability is at least $2/3$ if participants are ``semi-honest".
	With the definition of cheating used in~\cite{CGS13}, with ``semi-honest" participants, this bound is tight.
	However, the best known OT protocol has a cheating probability of $0.75$ if parties are not assumed to be semi-honest~\cite{CKS10}, meaning that there is a gap between what is known to be achievable, and what is known to be impossible. Narrowing this gap either way -- obtaining higher and thus tighter lower bounds on cheating probabilities, or finding concrete protocols with smaller cheating probabilities, leading to lower upper bounds -- is the main target of this paper.
	In order to obtain lower upper bounds, we consider general classes of protocols (either completely general or with some restrictions), but limit the capabilities of adversaries. This therefore provides only lower bounds on cheating probabilities, applicable to \emph{all} protocols within the considered class. To obtain upper bounds on cheating probabilities, we give a specific protocol, and then consider the most general attacks. This therefore provides an upper bound on achievable cheating probabilities, in the sense that the best protocol can perform at least as well as the specific protocol we give.
	There is also a subtlety regarding the requirement of semi-honesty, and related to this, to what extent dishonest parties can always obtain the information they would have obtained if they had been honest especially when considering variants of oblivious transfer that are not deterministic. We will return to this below.
	
	Our paper contains four main contributions:
	\begin{enumerate}
		
		\item We introduce the concept of Semi-random OT and prove a functional equivalence with respect to the cheating probabilities between 1-2 OT and Semi-random OT. We further describe a general framework for Semi-random OT. 
		\item We use this framework to show that the minimum achievable bound on the cheating probability is $2/3$. This agrees with the result in \cite{CGS13} for regular (deterministic) oblivious transfer, but in our case we  do not assume that parties are semi-honest.
		We also increase the lower bound on the minimum achievable cheating probability for 1-2 quantum OT protocols to $0.749$ if the states in the final round of the protocol when the parties are honest are pure and symmetric. We parametrise Alice's and Bob's ability to cheat in terms of a single variable $F$, related to the fidelity of the protocol output states.
		This parametrisation suggests how to construct schemes when either sender or receiver dishonesty is prioritized. \ That is, sender and receiver can have different cheating probabilities, and one can derive bounds for such situations.
		Such a scenario arises in the context of quantum signature schemes \cite{WDKA15,AWKA16}, and the derived bounds may prove useful for understanding the potential application of imperfect OT to signatures.
		
		\item We illustrate our construction by giving an OT protocol relying on unambiguous state elimination (USE) measurements. The protocol improves on previous protocols in the sense that it decreases the cheating probability of the receiver and is easier to implement. It also highlights the connection between USE measurements and 1-2 OT, and provides a new application for this relatively seldom used type of measurement. The security parameters achieved are almost tight with the bounds for protocols using pure symmetric states proven in this paper. In this protocol, one party has a smaller cheating probability than the other. This is not captured by the overall cheating probability, defined as the maximum of the cheating probabilities of either party. Such protocols might however be used for applications where restricting cheating by one party is prioritised. Such a protocol can also
		be combined with a ``trivial" protocol, to achieve a protocol with lower average cheating probability, where both sender and receiver can cheat with probability at most 0.74.  This is lower than the bound for protocols using pure symmetric states and constitutes an improvement on previously known protocols.
		
	\item Last, but not least, we present an optical realisation of the protocol we have given. In principle, an implementation of the protocol needs only the same components used for standard BB84 quantum key distribution. Each of the two qubits can be encoded into a single photon, sent individually to Bob, and measured using the same components as in BB84 quantum key distribution. That is, to implement our protocol, one only needs the components used for standard quantum key distribution. Our setup is however slightly different, because we want not only to test the protocol with the honest parties, but also experimentally implement the optimal cheating strategies and verify the predicted cheating probabilities. It is obvious that for any (e.g. commercial) application, the evaluation of the feasibility/practicality of the protocol considers the components required for an honest execution. Realizing cheating strategies is still of interest to evaluate how secure is the protocol in practice (c.f. quantum hacking). To implement these optimal cheating strategies 	requires usage of a nontrivial entangled state. We therefore encode two qubits into a single photon, and employ a linear optical quantum gate to prepare an entangled state where these two qubits are entangled with a third qubit retained by Alice, which is encoded in a separate second photon. The experimental results for both honest and cheating parties agree well with theoretical values, demonstrating that the protocol is feasible also when realised in this way.
		
	\end{enumerate}
	
	The paper is organised as follows. We begin in Section \ref{sec:defn} by defining 1-2 OT and Semi-random OT, stating an equivalence between the cheating probabilities for each. 
	In Section \ref{sec:general} we describe a general framework for Semi-random OT protocols and consider specific undetectable cheating strategies always available to Alice and Bob. We analyse these strategies to lower bound the achievable cheating probabilities for unbounded adversaries in 1-2 OT. In Section \ref{sec:use} we first introduce unambiguous measurements, in particular unambiguous state elimination (USE) measurements, and motivate their use in cryptography. We describe a semi-random OT protocol which employs USE measurements and analyse its security in the asymptotic limit. In Section \ref{sec:exp}, we present the experimental implementation of this protocol.

	\section{Definitions} \label{sec:defn}
	Intuitively, 1-2 OT is a two-party protocol in which Alice chooses two input bits, $x_0$ and $x_1$, and Bob chooses a single input bit $b$. The protocol outputs $x_b$ to Bob with the guarantees that Alice does not know $b$, and that Bob does not know $x_{b\oplus 1}$. A cheating Alice aims to find the value of $b$, while a cheating Bob aims to correctly guess both $x_0$ and $x_1$.
	
	At this point it is worth stressing that whenever we speak of cheating probability of one party, we assume that the other party executes the protocol honestly. This is a standard assumption in all cryptographic protocols with two competing parties (such as coin flip, bit commitment, and all versions of oblivious transfer) and we will adapt it in all the paper. The reason for this assumption is twofold. Firstly one is interested in ensuring that the ``interests'' of honest parties are secured, while it is less relevant to give guarantees to a cheating party. The second reason is that even defining what constitutes a cheating requires the other party to behave (at least to a point) honestly. For example, how can Bob cheat (guessing both $x_0$ and $x_1$) if Alice has not even chosen two bits?
	
	\begin{defn}{\emph{\cite{CKS10}}} \label{def:OT}
		A 1-2 quantum OT protocol is a protocol between two parties, Alice and Bob, such that
		\begin{itemize}[leftmargin=*]
			\item Alice has inputs $x_0, x_1\in \{0,1\}$ and Bob has input $b\in \{0,1\}$. At the beginning of the protocol, Alice has no information about $b$ and Bob has no information about $(x_0, x_1)$.
			\item At the end of the protocol, Bob outputs $y$ or Abort and Alice can either Abort or not.
			\item If Alice and Bob are honest, they never Abort, $y = x_b$, Alice has no information about $b$ and Bob has no information about $x_{b\oplus 1}$.
			\item $A_{OT} := \sup \{\text{\emph{Pr}}[\text{Alice correctly guesses $b ~\land$ Bob does}$ \\ 
			$\text{ \hspace*{1cm} not Abort}]\}$\\ \hspace*{0.8cm}  $= \frac{1}{2} +\epsilon_A$
			\item $B_{OT} := \sup \{\text{\emph{Pr}}[\text{Bob correctly guesses $(x_0,x_1)~\land$ Alice}$ \\ 
			$\text{ \hspace*{1cm} does not Abort}]\} $ \\ \hspace*{0.8cm} $= \frac{1}{2} +\epsilon_B$
		\end{itemize}
	\end{defn}
	The suprema are taken over all cheating strategies available to Alice and Bob. We note that there are also less common variants of the definition of $B_{OT}$, all with subtly different cheating implications. Ref. \cite{SCK14} defines cheating in terms of Bob being able to guess the XOR of Alice's bits, while Ref. \cite{CGS13} defines cheating in terms of Bob's ability to guess both bits, while also requiring that Bob can always retrieve a single bit with certainty. The choice of which definition is most appropriate will be largely application dependent. 
	
	We define $p_C := \max\{A_{OT}, B_{OT}\}$ to be the \textit{cheating probability} of the protocol. The maximum cheating probability characterises the performance of an OT protocol since protocols with $(A_{OT}=1,B_{OT}=0.5)$ are easy to construct. However, for certain applications, keeping track of cheating probabilities for both parties may be relevant. For example, it is conceivable that there are applications for which a protocol with cheating probabilities $(0.76, 0.5)$ may be better than one with $(0.75,0.75)$, and that protocols with same maximum cheating probability could be ordered with respect to the smaller cheating probability. Note also that our definition of security, while commonly used, differs from that in some other works, for example \cite{SSS09}, where security is characterised in terms of the information leakage, or in terms of Bob's ability to guess the output of some function $f(x_0,x_1)$. Nevertheless, our simpler definition makes sense if we are interested only in lower bounds on the cheating probability, since the ability to guess $(x_0,x_1)$ automatically implies the ability to guess $f(x_0,x_1)$ for any $f$.
	
	In this paper we define a variant of OT, Semi-random OT, which differs from the above 1-2 OT in that Bob does not have any inputs and randomly obtains one of Alice's bit values. More concretely, Semi-random OT is defined below.
	\begin{defn}
		1-2 quantum Semi-random OT, or simply Semi-random OT, is a protocol between two parties, Alice and Bob, such that
		\begin{itemize}[leftmargin=*]
			\item Alice chooses two input bits $(x_0, x_1)\in \{ 0, 1 \}$ or Abort.
			\item Bob outputs two bits $(c, y)$ or Abort.
			\item If Alice and Bob are honest, they never Abort, $y=x_c$, Alice has no information about $c$ and Bob has no information on $x_{c\oplus 1}$. Further, $x_0, x_1$ and $c$ are uniformly random bits \footnote{As also stated in \cite{CGS13}, it is standard to assume, in this setting, that an honest party's input bits are uniformly random, so that the corresponding cheating probabilities are computed on average.}. 
			
			\item $A_{OT} := \sup \{\text{\emph{Pr}}[\text{Alice correctly guesses $c ~\land$ Bob does}$ \\ 
			$\text{ \hspace*{1cm} not Abort}]\}$\\ \hspace*{0.8cm}  $= \frac{1}{2} +\epsilon_A$
			\item $B_{OT} := \sup \{\text{\emph{Pr}}[\text{Bob correctly guesses $(x_0,x_1)~\land$ Alice}$ \\ 
			$\text{ \hspace*{1cm} does not Abort}]\} $ \\ \hspace*{0.8cm} $= \frac{1}{2} +\epsilon_B$

		\end{itemize}
	\end{defn}
	
	The reason for introducing Semi-random OT is that we have found it simpler to work with than 1-2 OT, and the ability to perform Semi-random OT with cheating probabilities $A_{OT}$ and $B_{OT}$ implies being able to perform 1-2 quantum OT with the same cheating probabilities using additional classical communication and processing (See Appendix \ref{app:equiv}).
	Moreover, in spite of the equivalence in the above sense,
	semi-random protocols where Bob does not choose which bit he obtains can be subtly different from protocols where Bob can choose his input, in the following sense. In a semi-random protocol, such as the example protocol we give in section \ref{sec:use}, Bob obtains Alice's 1st or 2nd bit at random\footnote{A mechanism producing true randomness is a destructive quantum measurement.}. In other words, the protocol is not deterministic, even when parties honestly follow the protocol, and it generally involves a destructive quantum measurement. In order to obtain his ``honest'' output, Bob needs to irreversibly disturb the quantum state he possesses. In earlier papers~\cite{Lo97, CGS13} it is assumed, correctly for their framework, that Bob can always make a non-destructive measurement to obtain the bit of his choice. Bounds derived in this way then do not directly apply to Semi-random OT protocols, where such a measurement does not exist. Nevertheless, semi-random OT can be used to implement ``regular" OT, using classical post-processing as described in Appendix \ref{app:equiv}. There are subtle differences when considering how such post-processing affects lower and upper bounds on cheating. Here we directly obtain the same bound as in \cite{CGS13}, but by considering semi-random protocols. Our new technique also enables us to both increase the lower bound for protocols which use symmetric pure states, and to lower the upper bound by constructing a protocol with smaller cheating probabilities averaged over both parties.

	\section{Generic Protocol} \label{sec:general}

	In this section we introduce a general framework for Semi-random OT and use it to prove lower bounds on $p_C$. We present undetectable cheating strategies available to Alice and Bob and analyse them to lower bound their cheating probabilities $A_{OT}$ and $B_{OT}$ respectively. We show that for protocols within this framework, it holds that 
	\begin{equation}
	p_C = \max\{A_{OT},B_{OT}\} \geq 2/3.
	\end{equation}
	Further, if the states output to Bob by the protocol, when both parties are honest, are pure and symmetric, then
	\begin{equation}
	p_C = \max\{A_{OT},B_{OT}\} \gtrsim 0.749.
	\end{equation}
	We will prove this by bounding Alice's and Bob's cheating probabilities with respect to a single parameter, $F$, which is related to the fidelity of the output states of the protocol when it is honestly executed. (When either of the parties are dishonest, the output states may naturally be different.)
	From this we find that there is always a trade-off; as Alice's ability to cheat decreases, Bob's ability increases, and vice versa. 
	
	For this special case of pure symmetric output states, our result can be improved, giving an increased lower bound on the cheating probabilities. For protocols with pure symmetric output states, this nearly closes the gap between the known lower bounds, and the upper bounds resulting from existing protocols. We note that all 1-2 OT protocols we have seen proposed have output states that are pure and symmetric. Although there is no reason why this must be the case in general, protocols would intuitively often have this property. As we will later show, however, there exist protocols with lower average cheating probabilities than what is possible for protocols where the output states are pure and symmetric.

	\subsection{Protocol Framework}
	\label{sec:framework}
	
	We now describe the general framework for Semi-random OT protocols with $N$ rounds of communication between Alice and Bob. This framework is based on Kitaev's construction for strong coin flipping~\cite{Kitaev} and is useful for analysing the security of Semi-random OT.
	In Appendix \ref{app:equiv}, we further motivate why this framework is general for Semi-random OT.
	
	\begin{enumerate}
		\item Bob starts with the state $\rho_{BM}$ and Alice starts with an auxiliary system $A$ initialised to $\density{0}_A$. The overall state is $\rho_{BMA} := \rho_{BM}\otimes \density{0}_A$. We further suppose that Alice and Bob share the counter variable $i$, initialised to $1$, which tracks the round number of the protocol.
		\item Alice randomly selects an element $x_0x_1 \in \{ 00, 01, 11, 10 \}$.
		\item Bob sends system $M$ to Alice.
		\item Based on her choice in Step 2, Alice performs the unitary operation $U^{x_0x_1,i}_{MA} \in  \{U^{00,i}_{MA}, U^{01,i}_{MA}, U^{11,i}_{MA}, U^{10,i}_{MA}\}$. 
		\item Alice sends system $M$ back to Bob.
		\item Bob performs the unitary operation $V^{(i)}_{BM}$.
		\item The index $i$ is incremented by $1$. If $i=N+1$, the protocol proceeds to Step 8, otherwise it returns to Step 3.
		\item The final output held by Bob is 
		\beq \label{eq:output}
		\sigma^{x_0x_1}_{BM} := \Tr_A(\eta^{x_0x_1}_{BMA}),
		\eeq
		where
		
		\beq
		\eta^{x_0x_1}_{BMA} := V^{(n)}_{BM}U^{x_0x_1,n}_{MA}\dots V^{(1)}_{BM}U^{x_0x_1,1}_{MA}\circ\rho_{BMA},
		\eeq
		and we have used the convention $U\circ \rho= U\rho \  U^\dagger$.
		\item Bob performs a positive operator-valued measurement (POVM) with elements $\{ \Pi^{0*}_{BM}, \Pi^{1*}_{BM}, \Pi^{*0}_{BM}, \Pi^{*1}_{BM} \}$ to obtain the value of $c$ and $x_c$. The position of the star ``$*$" determines the value of $c$, i.e. $c=0$ for $0*$ and $1*$, while $c=1$ for $*0$ and $*1$. The value of the ``non-star" entry is the actual value of $x_c$. For example, the outcome $\Pi^{1*}_{BM}$ denotes that $c=0$ and $x_0=1$.
	\end{enumerate}
	
	The steps of the framework above describes the actions of Alice and Bob if they are honest, together with the associated outputs, assuming that all measurements are deferred to the end. Of course, Alice's and Bob's actual actions may deviate from the honest protocol description if they are dishonest, but we will see that to obtain our lower bound, this framework is useful.
	
	\subsection{Alice and Bob both honest}
	For the protocol to be correct if both Alice and Bob are honest, we require the following conditions to hold:
	\beq \label{eq:c0}
	\text{For  $c=0$: } \:\: \Tr(\Pi^{j*}_{BM} \sigma^{kl}_{BM}) = 
	\begin{cases}
		1/2,& \text{if } j=k,       \\
		0, &  \text{if } j\neq k.
	\end{cases}
	\eeq
	\beq \label{eq:c1}
	\text{For  $c=1$: } \:\: \Tr(\Pi^{*j}_{BM} \sigma^{kl}_{BM}) = 
	\begin{cases}
		1/2,& \text{if } j=l,       \\
		0, &  \text{if } j\neq l.
	\end{cases}
	\eeq
	These conditions imply that Bob receives either one of Alice's two chosen bits with equal probability, and that the bit received by Bob is correct.
	
	\subsection{Security against Bob} \label{sec:genbob}
	
	We assume that Bob acts honestly throughout the protocol, until step 9, where he deviates in the final measurement. This is clearly not the most general way of cheating for Bob, but any cheating probability that Bob can achieve by cheating in this restricted way can also be achieved by an unrestricted Bob.
	We will therefore be able to derive a lower bound on Bob's general cheating probability.
	Bob, at the beginning of step 9 (measurement), then holds either  $\sigma^{00}_{BM}$, $\sigma^{01}_{BM}$, $\sigma^{11}_{BM}$, or $\sigma^{10}_{BM}$. In order to cheat, Bob wants to guess the exact value of $x_0$ \emph{and} $x_1$. That is, he wants to know which of the four $\sigma$ states he holds. To do this, his optimal strategy would be to perform a minimum-error measurement. 
	However, the minimum-error measurement will vary according to the states chosen by any specific implementation of Semi-random OT.
	Instead, to provide a lower bound on Bob's optimal cheating probability for {\em all} protocols described by the framework, we assume that Bob performs a Square Root Measurement (SRM) \cite{HW94}. This may not be his optimal strategy, but it is a valid cheating strategy, and a strategy that  Bob can employ without even
	being caught (since Alice has no way of knowing which measurement Bob performs).  Bob's cheating probability is then at least as large as the success probability of the SRM, which is bounded as~\cite{AM14}
	\begin{equation} \label{eq:srm}
	p_{\rm succ}^{\rm SRM}\geq 1 - \frac{1}{8} \sum_{jk\neq lm} F(\sigma^{jk}_{BM}, \sigma^{lm}_{BM}),
	\end{equation}
	where $jk, lm \in \{00,01,11,10\}$ and $F$ is the fidelity, defined as
	\begin{equation}
	F(\rho, \sigma) := \text{Tr}\left(\sqrt{\rho^{1/2}\sigma\rho^{1/2}}\right).
	\end{equation}
	Eqs. \eqref{eq:c0} and \eqref{eq:c1} imply that $F(\sigma^{jk}_{BM}, \sigma^{j\oplus 1, k\oplus 1}_{BM}) = 0$ (since these states can be perfectly distinguished).  Without loss of generality, suppose that $\sigma^{00}_{BM}$ and $\sigma^{01}_{BM}$ are the pair of states with the highest fidelity. Define
	\beq
	F := F(\sigma^{00}_{BM}, \sigma^{01}_{BM}).
	\eeq
	Then it follows that
	\begin{equation} \label{eq:bot}
	B_{OT} \geq 1-F.
	\end{equation}
	This result is limited somewhat by the bound on the success probability of the SRM for general states given in Eq. \eqref{eq:srm}. Placing restrictions on the output states of the protocol allows us to tighten this bound. In particular, if $\{ \sigma^{00}_{BM}$, $\sigma^{01}_{BM}$, $\sigma^{11}_{BM}$, $\sigma^{10}_{BM}\}$ forms a symmetric set~\footnote{Symmetric sets of states are ubiquitous in quantum information. In this context ``symmetric'' means that there exists a 
	unitary $U$ such that $U^4 = \mathbbm{1}$ and $\sigma^{00}_{BM} = U\circ\sigma^{01}_{BM}=U^2\circ\sigma^{11}_{BM} = U^3\circ\sigma^{10}_{BM}$.} of pure states for which  $0\le F \le 1/2$, then as we show in Appendix \ref{app:symmBOT}, Bob's SRM measurement is successful with probability~\cite{WDA14} 
	\begin{equation} \label{eq:pure}
	\tilde p_{\rm succ}^{\rm SRM} \ge \frac{1}{4} \left( 1+\frac{1}{2}\sqrt{1-2F} + \frac{1}{2}\sqrt{1+2F}\right)^2,
	\end{equation}
	which gives the tighter bound $B^{\text{pure}}_{OT} = \tilde p_{\rm succ}^{\rm SRM} \geq\text{ the RHS of Eq.~(\ref{eq:pure})}$. (As we will see below, $F>1/2$ would mean that Alice's cheating probability is greater than 3/4.)
	
	If Bob's ability to cheat  does not depend on Alice's random choice of input, it seems likely that most protocols would output symmetric states, and this tighter bound would apply. 
	However, the example protocol we present in section \ref{sec:use}, which uses symmetric pure states, can be combined with a trivial protocol, to obtain overall average cheating probabilities which are lower than the bound for protocols using symmetric pure states. This shows that interestingly, protocols using symmetric pure states are not optimal for Semi-Random OT in general.
	
	\subsection{Security against Alice} \label{sec:genalice}
	Suppose Alice is dishonest and aims to guess the value of $c$ output to Bob. In this section we present a cheating strategy that is always available to Alice, and which is always undetectable. We derive Alice's cheating probability given that she performs this specific strategy, and use this to obtain a lower bound for Alice's achievable cheating probability given that she performs some optimal strategy, in the same way we restricted Bob's attacks to obtain a lower bound for his cheating probability.
	
	The strategy that Alice employs intuitively does the following. She chooses the two classical two-bit inputs that correspond to the pair of states among the $\sigma^{jk}_{BM}$ with the highest fidelity, which we called $F$ above. Then she performs the protocol operations corresponding to either classical input, conditioned on an ancillary qubit which is prepared in a superposition state, and which she keeps. In other words, the global state (before Bob's measurement) will be an entangled superposition, involving the pair of output states $\sigma^{jk}_{BM}$ with the highest fidelity on Bob's side. Bob then makes the measurement he makes if honest. Conditioned on his outcome, Alice's ancillary qubit is prepared in one of two states. Alice can distinguish between the two states with a success probability determined by the fidelity $F$ between the two states on Bob's side. (Her success probability is greater than 1/2, which would correspond to a random guess by Alice.) This leads us to a bound on Alice's cheating probability that involves the same quantity $F$ as our bound on Bob's cheating probability.
	
	More specifically, Alice can proceed as follows. Let $\ket{\Psi}_{BMAE}$ be a purification of $\rho_{BMA}$, where $E$ denotes the environment. Alice also prepares an additional state $|+\rangle_D=(\ket 0_D +\ket 1_D)/\sqrt 2$ for use as a control qubit to perform her strategy. Since we consider information-theoretic security, Alice can do anything allowed within quantum mechanics, including this. The overall state is
	\beq
	\frac{1}{\sqrt{2}}\left(\ket{\Psi}_{BMAE}\ket{0}_D + \ket{\Psi}_{BMAE}\ket{1}_D \right),
	\eeq
	with Alice in complete control of systems $A$, $E$ and $D$. Without loss of generality, we again assume that the two $\sigma$ states with the highest fidelity are $\sigma^{00}_{BM}$ and $\sigma^{01}_{BM}$. A valid cheating strategy available to Alice is as follows. In each Step 4 of the protocol, rather than performing a unitary $U^{x_0x_1,i}_{MA}$, Alice instead performs
	\begin{equation} \label{eq:alicestrat}
	U^{00,i}_{MA} \otimes |0\rangle\langle 0|_D + U^{01,i}_{MA} \otimes |1\rangle\langle 1|_D.
	\end{equation} 
	Defining Alice's overall operations as  $\mathcal{U} = V^{(N)}_{BM}U^{00,N}_{MA} \dots V^{(1)}_{BM}U^{00,1}_{MA}$ and $\mathcal{V} = V^{(N)}_{BM}U^{01,N}_{MA} \dots V^{(1)}_{BM}U^{01,1}_{MA}$, Alice's strategy leads to an output state
	\beq
	\begin{split}
		\ket{\chi} &:= \frac{1}{\sqrt{2}}\left(\mathcal{U}\ket{\Psi}_{BMAE}\ket{0}_D + \mathcal{V}\ket{\Psi}_{BMAE}\ket{1}_D \right) \\
		&:= \frac{1}{\sqrt{2}}\left(\ket{\psi^{00}}_{BMAE}\ket{0}_D + \ket{\psi^{01}}_{BMAE}\ket{1}_D \right).
	\end{split}
	\eeq
	This strategy is not detectable by Bob, since without access to system $D$ it is as if Alice has performed the honest operations for either $x=00$ or $x=01$, each with probability $1/2$. The states $\ket{\psi^{jk}}$ are purifications of $\sigma^{jk}_{BM}$, and all purifications are related by a unitary operation acting on the purifying system alone. Alice further performs the unitary operation 
	\beq
	W^{(1)}_{AE} \otimes \density{0}_D + W^{(2)}_{AE} \otimes \density{1}_D,
	\eeq
	where $W^{(1)}_{AE}$ and $W^{(2)}_{AE}$ are chosen to transform $\ket{\psi^{00}}$ and $\ket{\psi^{01}}$ into $\ket{\phi^{00}}$ and $\ket{\phi^{01}}$, such that the latter two states are the purifications of $\sigma^{00}_{BM}$ and $\sigma^{01}_{BM}$ with the highest overlap. 
	This operation is performed so that we can later use Uhlmann's theorem to express Alice's cheating probability in terms of $F$, as we shall see. The resulting state is
	\beq \label{eq:15}
	\ket{\Phi} := \frac{1}{\sqrt{2}}\left(\ket{\phi^{00}}_{BMAE}\ket{0}_D + \ket{\phi^{01}}_{BMAE}\ket{1}_D \right).
	\eeq
	
	In Step 8 of the protocol, Bob performs the POVM $\{ \Pi^{z}_{BM} \}_z$ on $\ket{\Phi}$, where $z\in\{0*,1*,*0,*1\}$. Our aim is to discover how well Alice can distinguish between the outcomes $c=0$ and $c=1$ using a measurement on her $D$ system. The state of system $D$ following Bob's POVM is
	\beq
	\mu_D = \frac{1}{2}\sum_{i,j,z}\bra{\phi^{0i}} \Pi^{z}_{MB} \ket{\phi^{0j}} \ket{j}\bra{i}_D,
	\eeq 
	where $i,j \in \{0,1\}$, $z \in \{ 0*, 1*, *0, *1 \}$. 
	
	Eqs. \eqref{eq:c0} and \eqref{eq:c1} can be used to evaluate terms of the form $\langle\phi^{jk} | \Pi^{z}_{BM}|\phi^{jk} \rangle$, since
	\beq \label{eq:17}
	\begin{split}
		\langle\phi^{jk} | \Pi^{z}_{BM}|\phi^{jk} \rangle &= \Tr_{BMAE}\Big(\Pi^{z}_{BM}\density{\phi^{jk}}\Big)\\
		& = \Tr_{BM}(\Pi^{z}_{BM}\sigma^{jk}_{BM}).
	\end{split}
	\eeq
	The expression for $\mu_D$ can be further simplified using the following lemma.
	\begin{lemma}
		For all values of  $z \in \{ 0*, 1*, *0, *1 \}$ and $jk \in \{ 00, 01,11,10 \}$ such that $\emph{\Tr}_{BM}(\Pi^{z}_{BM}\sigma^{jk}_{BM}) = 0$, it holds that
		\beq 
		(\Pi^{z}_{BM}\otimes \mathbbm{1}_{AE})\ket{\phi^{jk}}_{BMAE}=0.
		\eeq
	\end{lemma}
	\begin{proof}
		Since $\Pi^{z}_{BM} \otimes \mathbbm{1}_{AE}$ is a positive semidefinite operator, we can write its spectral decomposition as
		\beq
		\Pi^{z}_{BM} \otimes \mathbbm{1}_{AE} = \sum_n c_n \density{c_n},
		\eeq
		where all $c_n$ are positive real numbers. Therefore, using Eq. \eqref{eq:17},
		\beq
		\begin{split}
			\Tr_{BM}(\Pi^{z}_{BM}\sigma^{jk}_{BM}) = 0 &\Rightarrow \langle \phi^{jk}| \Pi^{z}_{BM}\otimes \mathbbm{1}_{AE} |\phi^{jk}\rangle = 0 \\
			&\Rightarrow \langle c_i | \phi^{jk} \rangle = 0 \:\: \forall i,
		\end{split}
		\eeq
		and the result follows.
	\end{proof}
	Using this lemma, $\mu_{D}$ simplifies to
	\begin{equation} \label{eq:24}
	\begin{split}
	\mu_{D} &= \frac{1}{2}\Bigg[ \frac{1}{2} |0\rangle\langle 0|_D + \langle \phi^{01}|\Pi^{0*}_{MB}|\phi^{00}\rangle |0\rangle\langle 1|_D + \\
	& \quad \:\: \langle \phi^{00}|\Pi^{0*}_{MB}|\phi^{01}\rangle |1\rangle\langle 0|_D + \frac{1}{2} |1\rangle\langle 1|_D \Bigg] \\
	& + \frac{1}{2}\Bigg[\frac{1}{2} |0\rangle\langle 0|_D + \frac{1}{2} |1\rangle\langle 1|_D \Bigg] \\
	& = \frac{1}{2} \mu^{c=0}_D + \frac{1}{2} \mu^{c=1}_D,
	\end{split}
	\end{equation}
	where the first square bracket corresponds to Bob obtaining an outcome $c=0$ (i.e. $\Pi^{0*}$ or $\Pi^{1*}$) and the second square bracket corresponds to Bob obtaining an outcome $c=1$  (i.e. $\Pi^{*0}$ or $\Pi^{*1}$). Lastly, we must evaluate $\langle \phi^{01}|\Pi^{0*}_{MB}|\phi^{00}\rangle$.
	
	To satisfy no-signalling, the density matrix in system $D$ must be the same regardless of whether or not Bob actually performs his measurement~\cite{rimini, bussey, jordan, nosig1, nosig2}. If Bob performs no measurement, using Eq. \eqref{eq:15}, the state of system $D$ is
	\begin{equation} \label{eq:nomeas}
	\begin{split}
	\frac{1}{2} [ & |0\rangle\langle 0|_D + \langle\phi^{01}|\phi^{00}\rangle |0\rangle\langle 1|_D \\
	& +  \langle\phi^{00}|\phi^{01}\rangle |1\rangle\langle 0|_D +  |1\rangle\langle 1|_D ].
	\end{split}
	\end{equation}
	Comparing Eqs. \eqref{eq:24} and \eqref{eq:nomeas}, we must have $\langle \phi^{01}|\Pi^{0*}_{MB}|\phi^{00}\rangle = \langle \phi^{01}|\phi^{00}\rangle$. The trace distance between $\mu^{c=0}_D$ and $\mu^{c=1}_D$ is therefore $|\langle \phi^{01}|\phi^{00}\rangle|$, meaning that Alice can distinguish $c=0$ from $c=1$ with probability 
	\begin{equation}
	\begin{split}
	p&= \frac{1}{2}\left(1 + |\langle \phi^{01}|\phi^{00}\rangle| \right) \\
	&= \frac{1}{2}\left(1 + F(\sigma^{00}_{BM}, \sigma^{01}_{BM}) \right) \\
	&:= \frac{1}{2}\left(1 + F \right) ,
	\end{split}
	\end{equation}
	where the second equality follows from Uhlmann's theorem \cite{Uhl76} since $\ket{\phi^{00}}$ and $\ket{\phi^{01}}$ are the purifications of $\sigma^{00}_{BM}$ and $\sigma^{01}_{BM}$ with maximum overlap. It therefore holds that 
	\begin{equation}
	\label{eq:aot}
	A_{OT} \ge  \frac{1}{2}\left(1 + F \right).
	\end{equation}
	
	\subsection{Result} \label{sec:res}
	Previously, the best known lower bound for the cheating probabilities in 1-2 quantum OT was~\cite{CGS13}
	\begin{equation} \label{eq:bound1}
	\max \{ A_{OT}, B_{OT} \} \geq 2/3.
	\end{equation}
	Our results in the previous section reproduce this bound since 
	\beq 
	\begin{split}
		A_{OT} &\ge \frac{1}{2}(1+F), \:\: B_{OT} \ge 1 - F \\
		& \Rightarrow  \:\: \min_{F} \Big( \max\{A_{OT}, B_{OT} \} \Big) = \frac{2}{3}.
	\end{split}
	\eeq
	Our way to obtain this bound differs substantially from \cite{CGS13} in two ways, and this means (as we will show later) that when imposing further restrictions on the class of protocols, we can increase the lower bound.

	If we consider protocols where the output states, during an honest execution, are pure and symmetric, then we obtain a tighter lower bound (which cannot be obtained using the technique in \cite{CGS13}). Specifically, 
	we can use Eq. \eqref{eq:pure} to obtain the tighter bound
	\begin{equation}
	\min_{F} \Big( \max \{ A_{OT}, B_{OT} \} \Big) \approx 0.749.
	\end{equation}
	Protocols using symmetric states may be preferrable due to theoretical or experimental simplicity, and intuitively, one might expect optimal protocols to employ symmetric states.
	
	Finally, another important feature of our bounding method is that our construction quantifies the trade-offs possible between $A_{OT}$ and $B_{OT}$, something of importance for applications where one is more interested in a smaller value for one of the two. This exact situation arises in the context of quantum signatures \cite{AWKA16}, where, in the distribution stage, signing keys are partially distributed in a manner reminiscent of 1-2 OT. In these protocols $A_{OT}$ is prioritised, and it is important that $A_{OT}\approx 0.5$ to protect against repudiation attempts. On the other hand, to protect against forging attempts is much simpler, and the requirements on $B_{OT}$ are less strict. The parametrisation of $A_{OT}$ in terms of $F$ suggests that in order to create an imperfect 1-2 OT schemes with a small  $\epsilon_A$, it is necessary to have a protocol which, in the honest case, outputs states that are almost orthogonal. Unfortunately, given $A_{OT} \approx 0.5$, our results show that it is necessary to have $B_{OT} \approx 1.$ This mirrors a similar result for two-party computation~\cite{Buhrman}.

	\section{A protocol for oblivious transfer} 
	\label{sec:use}
	In this section we present a protocol for imperfect quantum oblivious transfer which achieves cheating probabilities of 3/4 and approximately 0.729 for sender and receiver respectively. The protocol uses unambiguous quantum state elimination.
	
	\subsection{Unambiguous Measurements}
	
	Suppose that a quantum system is prepared in one of the states $\rho^x$, 
	where $x\in \mathcal{X}$, with prior probabilities $p_x$.
	When retrieving the information stored in $\rho^{x}$ using an ``optimal" measurement, what is ``optimal" depends heavily on the application. For communication protocols, 
	a minimum-error measurement -- one which identifies the state with the smallest probability of error -- is just one possibility. For cryptographic protocols, the optimal measurement is often one which returns the largest possible amount of information while simultaneously disturbing the system less than a threshold amount.
	
	A particular class of measurements we are interested in is unambiguous measurements. These measurements give ``perfect" information in the sense that, given a successful measurement outcome, one can be certain that the decoded classical information is correct. Unambiguous measurements come in two main flavours: unambiguous state discrimination (USD), and unambiguous state elimination (USE). A successful USD measurement on $\rho^{x}$ would identify $x$ with certainty, but the measurement is generally not successful with probability $1$. When the measurement is unsuccessful it does not uniquely determine the state. 
	
	USE measurements~\cite{PBR, Caves2002, Bando2014, Wallden2014, Heinosaari1, Perry, Heinosaari2, Havlicek, Crick} on the other hand can more often be successful with probability $1$, but only guarantee that $x\notin \mathcal{Y} \subset \mathcal{X}$, i.e. the measurement rules out states rather than definitively identifying the state. Intuitively, it seems that unambiguous measurements are well suited to cryptographic applications -- their ability to provide ``perfect yet partial" information on the states being sent is often exactly what is needed. More concretely, USD can be seen as very similar to Rabin OT, in which it is desired that the receiver obtains the sender's message with probability $1/2$, and otherwise receives nothing with probability $1/2$. On the other hand, USE measurements seem closely related to the more common 1-2 OT, in which incomplete but correct information is gained with certainty. Since OT plays a central role in secure two-party computation, it seems likely that unambiguous measurements could also play a role in this developing field. 
	
	\subsection{Semi-random OT using Unambiguous State Elimination}
	In this section, we present an application of USE measurements. We describe a protocol for implementing many runs of Semi-random OT and analyse its security in the asymptotic limit. We again work in the information--theoretic security setting but this time prove \textit{upper} bounds on the cheating probabilities achievable for Alice and Bob. We show that our protocol performs better than previous protocols, and is almost optimal with respect to the bounds for symmetric pure states derived in the previous section. The protocol proceeds as follows:
	\begin{enumerate}
		
		\item Alice uniformly, randomly and independently selects $\mathcal{N}$ elements from the set $\mathcal{X} = \{ 00, 01, 11, 10 \}$. She encodes elements as $00 \rightarrow |00\rangle$, $01 \rightarrow |++\rangle$, $11 \rightarrow |11\rangle$ and $10 \rightarrow |--\rangle$, where $|\pm\rangle = (|0\rangle \pm |1\rangle)/\sqrt{2}$.
		
		\item Alice sends the $\mathcal{N}$ two-qubit states to Bob.
		
		\item Bob randomly selects $\sqrt{\mathcal{N}}$ out of the $\mathcal{N}$ states he has received and asks Alice to reveal their identity~\footnote{The choice of $\sqrt{\mathcal{N}}$ test bits is somewhat arbitrary. For
			security in the asymptotic case, we only need Bob to choose a number of test states such that the number of test states tends to infinity as $N$ increases; the fraction of states chosen for testing tends to zero as $N$ increases.}. If Alice declares $|++\rangle$ or $|--\rangle$, then Bob measures both qubits in the $X$ basis, otherwise he measures both qubits in the $Z$ basis. The protocol aborts if any measurement result does not match Alice's declaration.
		
		\item The $\sqrt{\mathcal{N}}$ states used in the previous step are discarded.
		
		\item For each of the $\mathcal{N}-\sqrt{\mathcal{N}}$ remaining states, Bob measures the first qubit in the $Z$ basis and the second qubit in the $X$ basis. These measurements consitute two USE measurements (for example, an outcome of $\ket{0}$ on the first qubit rules out $\ket{11}$). Following these measurements, Bob can with certainty rule out one element from the set $\mathcal{Y}_0 = \{ 00, 11\}$, and one from the set $\mathcal{Y}_1 = \{01, 10\}$. In this way, for each of the remaining states he can know with certainty exactly one of $x_0$ and $x_1$, but not both.
	\end{enumerate}
	
	The result of this protocol is that Alice and Bob have performed $\mathcal{N}-\sqrt{\mathcal{N}}$ runs of Semi-random OT, each of which could be used to implement a single instance of 1-2 OT, as per the construction in Appendix \ref{app:equiv}. Below we analyse the cheating probabilities achieved by each instance of Semi-random OT generated by this protocol. 
	
	At this point it is important to note that in our analysis we assume that \emph{all} $\sqrt{\mathcal{N}}$ tests have passed successfully. This is important to simplify the subsequent analysis, by restricting to ``undetectable'' strategies as we will explain later. It is worth noting, however, that in realistic scenarios, even honest parties would fail some tests due to imperfections and noise. Therefore an important further work is to weaken the condition to allow for a small fraction of tests to fail, in order to make our protocol robust. This involves bounding the trace-distance of the resulting states as a function of the (small) failure of tests, and is postponed for a future publication.
	
	Note that, from a security perspective, the protocol given above can be set in the general framework considered of the previous section by defining $U=R\otimes R$, where
	\begin{equation}
	R = |+\rangle\langle 0| - |-\rangle\langle 1|.
	\end{equation}
	Alice begins with the state $|00\rangle$ and applies either $\mathbbm{1}$, $U$, $U^2$ or $U^3$ to obtain either $|00\rangle$, $|++\rangle$, $|11\rangle$ or $|--\rangle$ respectively. The subsequent rounds simply consist of classical communication and measurements, the latter of which can be described as a unitary operation acting on a larger Hilbert space, with state collapse delayed until a protocol output is required. We show that this protocol can be made secure with $A_{OT} = 0.75$ and $B_{OT} \approx 0.729$.

	\subsection{Security against Bob}
	
	If Bob wants to cheat, then his aim is to correctly guess both $x_0$ and $x_1$ for each individual pair. In the asymptotic limit, the fraction of states discarded for testing in Step 3 tends to zero. Since the states are prepared independently, any strategy Bob performs (including general measurements correlated across all $\mathcal{N}$ states) cannot have an \textit{average} success probability (probability of correctly identifying both $x_0$ and $x_1$) which is greater than the minimum-error measurement on a single state~\footnote{If there were such a measurement, Bob could simulate this strategy when he has only a single state and beat the minimum-error measurement.}.
	Therefore, in the asymptotic limit we can bound Bob's average 
	cheating probability for each of the $\mathcal{N} -\sqrt{\mathcal{N}} \approx \mathcal{N}$ runs by considering the minimum-error measurement on a single state. 
	Since the set $S:=\{|00\rangle, |++\rangle, |11\rangle, |--\rangle \}$ forms a set of symmetric pure states, the minimum-error measurement is the SRM \cite{WDA14}. Using this measurement Bob can guess both of Alice's input bits with probability 
	\begin{equation}
	B_{OT} = \frac{1}{4}\left( 1+\frac{1}{\sqrt{2}} \right)^2 \approx 0.729.
	\end{equation}
	In this case, Bob's optimal strategy is the exact strategy considered in the general scenario in Section \ref{sec:genbob}.
	(If the tested fraction of states does not tend to zero as $\mathcal{N}\rightarrow \infty$, then Bob's optimal measurement would be a maximum confidence measurement~\cite{maxconf, nosig2}, with a success probability increasing with the fraction of tested states, reaching a maximum of 3/4 if at least 1/4 of the states are tested. Bob would then perform the relevant measurement with higher confidence in the result, and if the measurement fails, ask to ``test" the state in that position.)

	\subsection{Security against Alice}
	If Alice wants to cheat, her aim is to correctly guess the value of $c$ such that Bob received $x_c$. To do this, she may send states other than the ones in $S$. In general, Alice will generate $\rho_{AB_{11}B_{12}B_{21}B_{22}...B_{N1}B_{N2}}$ and send the $B$ systems to Bob, keeping the $A$ system for herself. In Step 3 of the protocol Bob then randomly selects a pair of the qubits he received, say $\rho_{B_{k1}B_{k2}}$, and asks Alice to declare the identity of the state. He does this for $\sqrt{\mathcal{N}}$ of the $\mathcal{N}$ pairs. Since we are looking for an upper bound on Alice's capabilities, we assume that she holds a purification $|\Psi\rangle_{B_{k1}B_{k2}A}$ of $\rho_{B_{k1}B_{k2}}$. 
	
	Alice must declare a state to Bob that will agree with his measurement outcomes in Step 3. If she can do this with certainty, then the state $|\Psi\rangle_{B_{k1}B_{k2}A}$ must be of the form
	\begin{equation} \label{eq:state}
	\begin{split}
	|\Psi\rangle_{B_{k1}B_{k2}A} &= b_{0} |00\rangle_{B_{k1}B_{k2}}|0\rangle_A + b_{1} |++\rangle_{B_{k1}B_{k2}}|1\rangle_A \\
	&+  b_{2} |11\rangle_{B_{k1}B_{k2}}|2\rangle_A + b_{3} |--\rangle_{B_{k1}B_{k2}}|3\rangle_A,
	\end{split}
	\end{equation}
	where $\{|0\rangle_A, |1\rangle_A, |2\rangle_A, |3\rangle_A \}$ is an orthonormal basis. If Alice does not send states in the above form, then she cannot guess Bob's measurement outcomes with certainty, and for asymptotically large $\mathcal{N}$ it becomes virtually certain that the protocol will abort.
	
	We note that Alice also cannot improve her average cheating probability by using strategies where she uses entanglement not just between the system she keeps and Bob's individual qubit pairs, but where she also introduces entanglement between the different qubit pairs she sends to Bob. 
	Any state for which Alice will deterministically pass a test on the qubits in position $B_{k1}B_{k2}$, can be written as
	\begin{equation} \label{eq:entstate}
	\begin{split}
	|\Psi\rangle_{B_{k1}B_{k2}A'} &= b_{0} |00\rangle_{B_{k1}B_{k2}}|0\rangle_{A'} + b_{1} |++\rangle_{B_{k1}B_{k2}}|1\rangle_{A'} \\
	&+  b_{2} |11\rangle_{B_{k1}B_{k2}}|2\rangle_{A'} + b_{3} |--\rangle_{B_{k1}B_{k2}}|3\rangle_{A'},
	\end{split}
	\end{equation}
	where $\{|0\rangle_{A'}, |1\rangle_{A'}, |2\rangle_{A'}, |3\rangle_{A'} \}$ is an orthonormal basis which may include not just a system Alice holds, but Bob's qubits in other positions than $B_{k1}B_{k2}$. This state is evidently of the form in \eqref{eq:state}. That is, if Alice is able to deterministically pass a test done on a qubit pair, then this directly limits her average cheating probability for that qubit pair, and this is true for all qubit pairs also when Alice can entangle the qubits she sends to Bob in arbitrary ways.

	Essentially, this means that Alice is restricted to the attacks considered in the general protocol analysis in Section \ref{sec:genalice} -- attacks that are superpositions of honest operations, and as such are always undetectable by Bob. In fact, it can be proven (see Appendix \ref{app:alice}) that an optimal strategy for Alice is to prepare
	\begin{equation}
	\label{eq:optcheatstate}
	\frac{1}{\sqrt{2}} \left( |00\rangle_B|0\rangle_A +|++\rangle_B|1\rangle_A\right),
	\end{equation}
	which corresponds exactly to the operation given in Eq. \eqref{eq:alicestrat}. Since the overlap between all adjacent states in $S$ is $1/2$, Eq. \eqref{eq:aot} implies that Alice can correctly guess the value of $c$ with probability $3/4$. The analysis in Appendix \ref{app:alice} confirms that this is her cheating probability.

	\subsection{A combined protocol with lower average cheating probability}
	
	One can combine our example scheme, where $A_{OT}=3/4$ and $B_{OT}=0.729$, with a ``trivial" scheme where $A_{OT}=1/2$ and $B_{OT}=1$, to achieve a scheme where both Alice's and Bob's average cheating probabilities are below 3/4. Note that this is possible because our protocol had different cheating probabilities for sender and receiver. This illustrates that the maximum of the two cheating probabilities does not fully characterise the performance of a protocol, since the smaller cheating probability can become relevant in such combined protocols. As in~\cite{CGS13}, Alice and Bob execute a weak coin flipping protocol to probabilistically choose between a protocol that is more favourable to Alice, and one that is more favourable to Bob. In~\cite{CGS13}, it is considered in some detail how to securely compose weak coin flipping and a subsequent OT protocol. 
	In the trivial OT scheme we will use, Alice simply sends Bob both bits, and Bob reads the bit he wants and discards the other, giving $A_{OT}=1/2$ and $B_{OT}=1$. If our example scheme is chosen with probability $p$ and the trivial scheme chosen with probability $1-p$, the average cheating probabilities become
	\begin{equation}
	\tilde A_{OT} = 3p/4 + (1-p)/2,~~\tilde B_{OT} = 0.729 p + (1-p).
	\end{equation}
	Choosing $p$ to set these equal results in a combined scheme where both Alice and Bob can cheat on average at most with probability $\tilde A_{OT}=\tilde B_{OT}=p_C \approx0.74$. This is the smallest cheating probability that a concrete protocol can achieve to our knowledge. Interestingly, this is lower than 0.749 both for Alice and Bob, thus proving that protocols using symmetric pure states are not optimal for semi-random oblivious transfer in terms of average cheating probability.

	\section{Experiment}
	\label{sec:exp}
	A major advantage of the above protocol is that it can be realized using standard BB84 quantum key distribution setup \footnote{Actually, Bob doesn't need a quantum memory for his test measurements. He can randomly decide to make a test measurement in the $ZZ$ or the $XX$ basis, and only afterwards ask Alice to reveal the corresponding states. Half of the test measurements will not contribute, but this will not affect the function of the protocol.}. 
	However, we have implemented the semi-random OT protocol slightly differently to enable also the realization of optimal cheating strategies. Namely, we created the Alice's entangled state with the help of optical multi-qubit quantum logic gates. But still one photon carrying a single qubit stays at Alice's side and the other photon, carrying two qubits travel to Bob's side.
	
	\begin{figure*}[t]
		\centering
		\includegraphics[width=\textwidth]{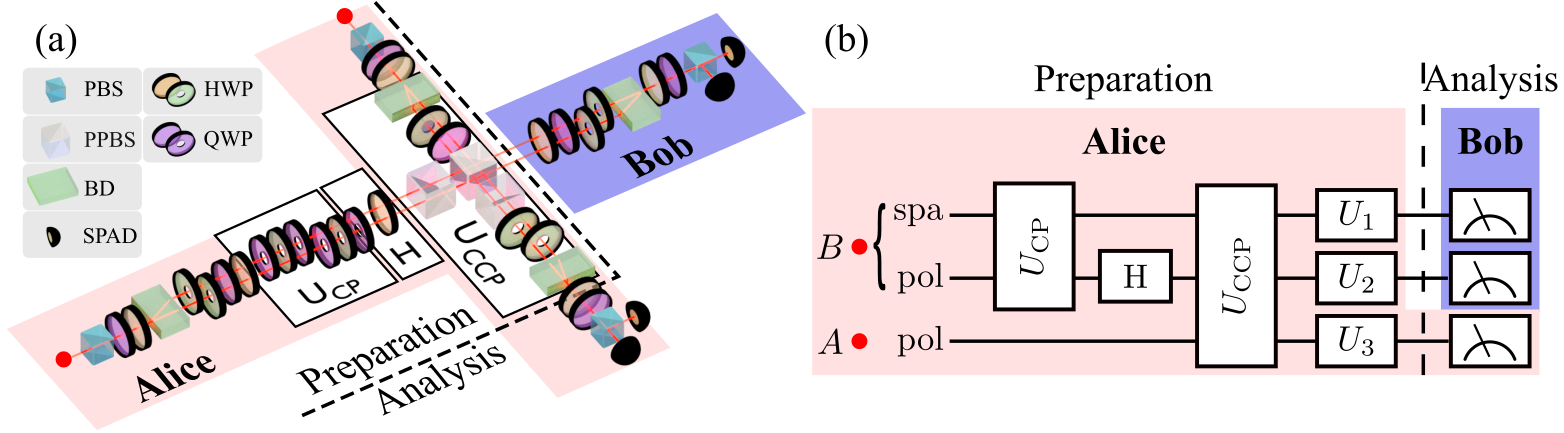}
		\caption{(a) - Experimental setup. (b) - Quantum circuit diagram of the experiment. With appropriate tuning of controlled-phase gates $U_{\mathrm{CP}}$, $U_{\mathrm{CCP}}$, and the single-qubit gates $U_{1,2,3}$, Alice prepares the required state (spa and pol denote qubits encoded into spatial and polarization modes respectively).}
		\label{fig:1}
	\end{figure*}

	\subsection{Experimental setup}
	Pairs of 810-nm time-correlated photons were generated using type-II spontaneous parametric down-conversion in a $\beta$-barium-borate crystal. The photons were guided to the experimental setup depicted in Figure~\ref{fig:1}a.
	Primarily, the state of the first of the qubits $B$ chosen by Alice was encoded by quarter- and half-wave plates (QWP, HWP) into the polarization of the signal photon. 
	Then a calcite beam displacer (BD) spatially separated horizontally and vertically polarized components into two parallel beams with a lateral distance of 4 mm. This turns the encoding of the first qubit from polarization to spatial encoding. Wave plates acting on both parallel beams were then used to encode the state of the second qubit $B$ into polarization. In this way, a single photon carried both qubits.
	
	When the basic operation of the semi-random OT was tested, as well as when Bob's cheating strategy was implemented, we utilized the idler photon (the other photon in the pair) only to herald successful generation of the signal photon.
	When Alice's cheating strategy was studied, the state of Alice's qubit $A$ was encoded into the polarization state of the idler photon.
	Linear-optical quantum logic gates, shown in Figure~\ref{fig:1}b, then entangled the input qubits to produce the required state (\ref{eq:optcheatstate}). 
	
	The two-qubit controlled-phase gate ($U_{CP}$) operates on qubits $B$ and introduces an arbitrary phase shift on state $|11\rangle$. The wave plates in the lower optical path perform the phase shift, the wave plates in the upper path only compensate for the path length difference. Another half-wave plate implements the Hadamard gate acting on the second one of qubits $B$ (encoded in the polarization degree of freedom). 
	The three-qubit controlled-controlled-phase gate ($U_{CCP}$) provides a way to entangle qubit $A$ with qubits $B$. The beam displacer separates the path of the idler photon according to its polarization into two parallel beams with 6-mm spacing. This extends the Hilbert space, providing room for manipulation. Suitable polarization operations, two-photon interference, and consecutive coincident detection then constitute the $U_{CCP}$ operation. The two-photon interference takes place in the central block of three partially-polarizing beam splitters (PPBS), the central one with reflectances $R_H = 0, R_V = 2/3$ the other two with $R_H = 2/3, R_V = 0$. This is the core of the gate operation \cite{Okamoto2005, Langford2005, Kiesel2005, Starek2016} which is explained in detail under Methods in our previous work \cite{Starek2020}. The gate is probabilistic and succeeds with theoretical probability 1/9 for phase shifts 0 and $\pi$, which are used in the experiment.
	
	Final projective measurements are realised by wave plates, polarizing beam splitters, and single-photon avalanche diodes (SPAD). This enables projection onto an arbitrary product state \footnote{We use a simplified configuration at Bob's side, but it is possible to build a four-output measurement spanning the full two-qubit space solely by linear optics.}. 
	Electric signals are processed by coincidence logic. The overall coincidence count rate was roughly 330 counts per second. The experimental integration time was 5~s for each projective-measurement setting.

	\subsection{Both parties are honest}
	To test the case when both parties are honest, we set the $U_{\mathrm{CP}}$ and $U_{\mathrm{CCP}}$ gates to zero phase shift and turned off the Hadamard operation $H$.
	
	We sequentially prepared states $|00\rangle$, $|++\rangle$, $|--\rangle$, $|11\rangle$ and measured each of them in $ZX$ basis on Bob's side. The probability of Bob correctly receiving one of Alice's bits was estimated to be 0.9943(9), where the number in the brackets represents one standard deviation at the final decimal place.
	It means that due to experimental imperfections, there is a small probability (about 0.6\%) that Bob obtains an erroneous bit value. Complete experimental data are provided in Table~\ref{tab:CorrectTransfer} of Appendix~\ref{app:expdata}.  
	
	The protocol also includes test measurements. If the parties are honest, this means that the states $|00\rangle$, $|11\rangle$ are measured in the $ZZ$ basis and states $|++\rangle$, $|--\rangle$ in the $XX$ basis. Such measurements should unambiguously discriminate between the incoming states and Bob should never abort the protocol when Alice is honest. But in an experimental implementation imperfections may cause errors. In our experiment, the average error probability was 0.013(1). All measured data are provided in Table~\ref{tab:HonestAlarms} of Appendix~\ref{app:expdata}.

	\subsection{Bob is cheating}
	Bob's optimal cheating strategy is to perform a minimum-error measurement \footnote{In this situation it is a square-root measurement, which is actually quite intuitive: the states $\ket{0}, \ket{1}, \ket{+}, \ket{-}$ form a ``cross'' on the Bloch sphere. The measurement on the first qubit is represented by two orthogonal states which lie on the diagonal. The measurement on the second qubit is different and corresponds to the other diagonal.}. 
	In our case, this means measuring the first qubit in the basis 
	$$\{|\zeta_{0}\rangle = \alpha\ket{0} + \beta\ket{1}, |\zeta_{1}\rangle = \beta\ket{0} - \alpha\ket{1} \}$$ 
	and the other in the basis 
	$$\{|\xi_{0}\rangle = \alpha\ket{0} - \beta\ket{1}, |\xi_{1}\rangle = \beta\ket{0} + \alpha\ket{1} \}$$
	with $\alpha = \cos(\pi/8)$ and $\beta = \sin(\pi/8)$. Each combination of detector clicks gives Bob a guess of both the Alice's bits. The average experimental value of cheating probability, i.e., the probability of a correct guess of both bits, was $0.718(5)$, which is close to the theoretical value of $0.729$. Recorded counts are provided in Table~\ref{tab:BobCheating} of Appendix~\ref{app:expdata}. 
	
	\subsection{Alice is cheating}
	To test Alice's optimal cheating strategy, we set the phase shifts of the gates $U_{\mathrm{CP}}$ and $U_{\mathrm{CCP}}$ to $-138.2^{\circ}$ and $180^{\circ}$, respectively. We prepared the input qubits in a suitable product state and adjusted the output single-qubit operations $U_{1,2,3}$ to achieve the desired entangled state (\ref{eq:optcheatstate}). The specific choice of input states, gate parameters, and unitary operations is a result of numerical optimization, which is discussed in Appendix \ref{app:cheatstateprep}.
	
	In order to verify the prepared entangled state, we performed quantum state tomography \cite{Jezek2003}. The purity of the state was $P = 0.884$ and its fidelity with respect to the ideal state (\ref{eq:optcheatstate}) was $F = 0.921$. The cause of imperfect purity and fidelity is the sensitivity of the $U_{\mathrm{CCP}}$ gate to interferometric phase instability and spatio-temporal misalignment of the photons. Imperfect wave-plate retardances reduce the quality of the state even further.
	
	To learn which bit was obtained by Bob, Alice measures her qubit $A$ in the state (\ref{eq:optcheatstate}) in the $X$ basis. Honest Bob makes his measurements according to the protocol. As described above, Bob's outcomes $\ket{0+}_B, \ket{1-}_B$ correspond to $c=0$ and $\ket{1+}_B, \ket{0-}_B$ correspond to $c=1$. If Alice obtains $\ket{+}_A$ ($\ket{-}_A$), then she guesses that $c=0$ ($c=1)$. Alice's measurements in the $X$ basis and Bob's measurements in the $ZX$ basis were already contained in the data from the three-qubit state tomography. We estimated the cheating probability as the number of detection events in which Bob and Alice obtain the same value of $c$, divided by the number of all detection events. Alice correctly estimated Bob's bit $c$ with probability $0.77(1)$. The measured count rates are in Table~\ref{tab:AliceCheating1} in Appendix~\ref{app:expdata}. 
	
	In the case of test measurements, Bob measures in the $ZZ$ or the $XX$ basis and Alice in the $Z$ basis. These data were also obtainable from the tomographic measurement. In theory, Bob should not be able to detect this type of cheating strategy by Alice. But in the experiment, there was a small fraction of outcomes telling Bob to abort the protocol, on average 0.059(6). This fraction was calculated as the number of counts in which Bob's measurement outcome did not match Alice's declaration divided by the total number of counts. The relevant data are presented in Table~\ref{tab:AliceCheating2} of Appendix~\ref{app:expdata}. 
	
	In our experiment, Alice's probability of making a correct guess, $0.77$, was higher than the theoretical limit $0.75$. But there was also a relatively high probability of Bob discovering her cheating (0.059, which is higher than the probability of ``false alarm'', 0.013, if Alice was honest). These effects are likely caused by imperfect preparation of the state (\ref{eq:optcheatstate}).

	\section{Discussion}

	In this paper we introduced Semi-random oblivious transfer (OT) and a general framework useful for its study.
	We explicitly constructed undetectable cheating strategies available to Alice and Bob and used them to lower-bound the cheating probability for any Semi-random OT protocol within our framework. The derived bounds are directly transferable to standard 1-2 quantum OT, allowing us to obtain the lower bound $p_C \geq 2/3$, but using different assumptions on cheating strategies than assuming semi-honest adversaries as done by  Chailloux {\em et al.}~\cite{CGS13}. 
	Our technique, other than re-deriving the previous bound, allows us to (i) quantify the trade-off between cheating probabilities for different parties, which can be useful for applications where limiting cheating by one party is prioritised and (ii) obtain tighter bounds if we impose further restrictions. In particular, if the states used by honest parties are pure and symmetric, we obtain the bound $p_C \geq 0.749$, which was not obtained previously.

	Our construction provides a simple quantitative relationship between Alice's and Bob's ability to cheat, and gives new bounds in biased settings. In applications more sensitive to sender dishonesty than receiver dishonesty (or vice versa), our parametrisation of $A_{OT}$ and $B_{OT}$ in terms of the fidelity shows explicitly how reductions in one party's ability to cheat will impact the other's cheating probability. 
	To illustrate our construction we presented an OT protocol using
	unambiguous state elimination measurements to achieve cheating probabilities $A_{OT}=3/4$, $B_{OT} \approx 0.729$ and therefore $p_C = 3/4$, together with its experimental realisation. The cheating probabilities compare favourably with the previously best known protocol given in Ref. \cite{CKS10} in which $A_{OT}=B_{OT}=3/4$. Unlike for the qutrit protocol proposed in~\cite{CKS10}, in our example protocol, the bound on Alice's cheating probability concerns her average cheating probability. On the other hand, Bob's cheating probability is lower (0.729 against 0.75 in~\cite{CKS10}), and above all, our protocol does not require entanglement and can be realised using the same experimental components as BB84 quantum key distribution. A minor modification could render our protocol even more practical. Bob could, before asking Alice to reveal any states, randomly select some qubit pairs and measure them in the same basis, either the $X$ or the $Z$ basis. He then asks Alice to receive these states, but only after he has measured these qubit pairs. If Alice's declaration does not match his measurement results, he again aborts. Bob's test is then only useful if his selected basis matches the basis states used by Alice.
	Another variation would be for Bob to randomly select which qubit he measured in the $X$ and which in the $Z$ basis. This makes no difference if Alice is limited to using undetectable cheating strategies, but would lead to somewhat improved performance when loss and imperfections are present and in finite-size scenarios, where Alice may choose to employ a cheating strategy that could be detected by Bob with some probability.
	
	Since our example protocol outputs symmetric pure states, the cheating probabilities achieved are almost tight with the bounds proven in this paper for this class of protocols. 
	Combining the example protocol with a trivial protocol, however, an average cheating probability $p_C\approx 0.74$ for both Alice and Bob is possible. It follows that protocols with pure and symmetric output states are not optimal. There thus remains a gap between the known lower bounds on cheating probabilities for quantum oblivious transfer, and what the lowest achievable cheating probabilities are.
	
	We further note that if two protocols are combined using weak coin flipping, then the parties know which protocol actually got implemented. The bound on cheating probabilities in such combined protocols are therefore also only bounds on average cheating probabilities. For an individual round, the parties are aware that they have higher or lower cheating probabilities.
	Related to this, cheating probabilities do not fully capture how certain a cheating party can be that the extra information they have dishonestly obtained is correct. In our example protocol, Bob can never be certain that his dishonestly obtained information is correct. He only ever knows that his guess is correct with probability 0.729. Alice, however, can be certain of Bob's bit choice with probability 1/4, and she knows when this occurs. The rest of the time her guess is right with probability 2/3. This is a further advantage of our protocol, compared with the one in~\cite{CKS10}. To elaborate, if one probabilistically chooses between a trivial protocol where Alice can cheat perfectly and Bob cannot cheat at all ($A_{OT} = 1$ and $B_{OT} = 1/2$) and a trivial protocol where Alice cannot cheat at all and Bob can cheat perfectly ($A_{OT} = 1/2$ and $B_{OT} = 1$), then the average cheating probabilities for either party are 3/4, but with probability 1/2, either party knows for sure that they can cheat perfectly. When executing the protocol in~\cite{CKS10}, Alice similarly knows for sure what Bob's bit choice was half the time, and the rest of the time she randomly guesses. In our protocol, Alice is only sure with probability 1/4. Bob, however, cheats with a minimum-error measurement both in our protocol and the one in~\cite{CKS10}, and is never sure that his guess is correct. Since the states Bob receives in both protocols are linearly dependent, he can never unambiguouosly determine both of Alice's bit values. 
	We also presented an optical realisation of our protocol. The achieved experimental performance parameters agree well with the theoretical values, showing that the protocol is feasible.
	
	As a final point we note that in quantum cryptography, it is often easier to analyse so-called i.i.d. (individual identically distributed) cheating strategies, where dishonest parties are restricted to act individually on each quantum system transmitted (or to act individually on other relevant "units" in the protocol), and where they act in the same way for each transmitted quantum system. If the parties can use cheating strategies that operate jointly on several transmitted quantum systems, sometimes called "coherent" cheating strategies, then cheating probabilities might increase. It is therefore worth emphasizing that the results we have obtained are in fact valid for general cheating strategies, not just i.i.d. cheating strategies. First, note that the bounds we have derived are lower bounds for cheating probabilities, and therefore immediately are valid for all cheating strategies, including joint or coherent cheating strategies by either party. Second, in the example protocol, we did not need to restrict either Alice or Bob to i.i.d. cheating strategies. As for Alice, in connection with Eq. (\ref{eq:entstate}), we explained why she does not benefit from entanglement with other positions. That is, we are allowing her joint cheating strategies, and show that this does not increase her ability to predict Bob's output for each instance of OT. However, it should be pointed out that this results from the fact that we make the simplifying assumption that Alice needs to pass Bob's tests with unit probability. If this assumption is not made, then the analysis of whether joint or coherent strategies can help Alice cheat is less straightforward. If Alice is allowed to fail Bob's tests with some probability, then she can use a state that slightly deviates from the state in Eqs. (\ref{eq:state}) and (\ref{eq:entstate}), and a more careful analysis of i.i.d. versus joint or coherent cheating strategies for Alice would be required. Bob, on the other hand, needs to maximise his average probability to correctly guess both of Alice's bits. His optimal cheating probability is obtained by individual minimum-error measurements on each qubit pair. Joint measurements on more than one qubit pair do not help him, and there is no need to restrict Bob to i.i.d. cheating strategies in the finite-size scenario either.
	
	\acknowledgements{The authors would like to thank J. Sikora and I. Kerenidis for helpful discussions. This work was supported by the UK Engineering and Physical Sciences Research Council (EPSRC) under 		EP/T001011/1, EP/T001062/1 and EP/M013472/1.
		R.A. gratefully acknowledges EPSRC studentship funding under grant number EP/I007002/1.
		R.S., M.M., L.M., and M.D. acknowledge support by
		Palack\'{y} University under grant number IGA-PrF-2020-009.}

	\bibliographystyle{apsrev4-1}

	\appendix
	\section{Equivalence between Semi-random OT, OT, and Random OT}
	\label{app:equiv}
	Here we prove the following claim (stated below) contained in the main paper.
	
	\begin{proposition}
		The existence of a Semi-random OT protocol with cheating probabilities $A_{OT}$ and $B_{OT}$ is equivalent to the existence of a 1-2 quantum OT protocol with the same cheating probabilities.
	\end{proposition}
	
	To prove this, we begin by giving the definition of a related OT variant called Random OT (ROT), as follows.
	
	\begin{defn} {\em Random OT is a protocol between two parties, Alice and Bob, such that}
		
		\begin{itemize}[leftmargin=*]
			
			\item Alice outputs two bits $(x_0 , x_1) \in \{0, 1\}$ or Abort.
			
			\item Bob outputs two bits $(c,y)$ or Abort.
			
			\item If Alice and Bob are honest, they never Abort, $y = x_c$, Alice has no information about $c$ and Bob has no information about $x_{c\oplus 1}$. Further, $x_0,x_1$ and $c$ are uniformly random bits.
			
			\item $A_{OT} := \sup \{ \text{\emph{Pr}}[\text{Alice correctly guesses~} c~\land \text{Bob does }$ \\ \hspace*{1.1cm}   $\text{not Abort}]\}$\\ \hspace*{0.8cm}  $= \frac{1}{2} +\epsilon_A$
			\item $B_{OT} := \sup \{\text{\emph{Pr}}[\text{Bob correctly guesses $(x_0,x_1)~\land$ Alice}$ \\ $\text{ \hspace*{1cm} does not Abort}]\} $ \\ \hspace*{0.8cm} $= \frac{1}{2} +\epsilon_B$
			
		\end{itemize}
	\end{defn}
	
	Ref.~\cite{CKS10} proved that the existence of a  ROT protocol with cheating probabilities $A_{OT}$ and $B_{OT}$ is equivalent to the existence of a 1-2 OT with the same cheating probabilities.
	Following very similar arguments, in the following subsections we will show that the existence of a Semi-random OT protocol with cheating probabilities $A_{OT}$ and $B_{OT}$ is equivalent to the existence of a ROT with the same cheating probabilities. This, combined with the results in Ref.~\cite{CKS10}, proves the proposition.
	
	\subsection{Semi-random OT from ROT}
	\label{sec:SRfromROT}
	
	Let $P$ be a ROT protocol with cheating probabilities $A_{OT}(P)$ and $B_{OT}(P)$. We construct a Semi-random OT protocol $Q$ with the same cheating probabilities as follows:
	\begin{enumerate}
		\item Alice has inputs ($z_0, z_1)$.
		\item  Alice and Bob run protocol $P$ to output $(x_0, x_1)$ for Alice and $(c, y)$ for Bob.
		\item  Alice and Bob abort in $Q$ if and only if they abort in $P$. Otherwise, Alice sends $(z_0 \oplus x_0, z_1 \oplus x_1)$ to Bob.
		\item Bob outputs $(c, y')$ where $y' = (z_c \oplus x_c \oplus y)$.
	\end{enumerate}
	We now show that $Q$ is a Semi-random OT protocol with cheating probabilities $A_{OT}(P)$ and $B_{OT}(P)$.
	
	If Alice and Bob are honest, then by definition we have $ y=x_c$ and so $y' =z_c$. Alice has no information about $c$ and Bob has no information about $z_{c\oplus 1}$, as required.
	
	If Alice is dishonest, she cannot guess $c$ except with probability $A_{OT}(P)$ since she only receives communications from Bob via protocol $P$. Therefore $A_{OT}(Q)=A_{OT}(P)$.
	
	If Bob is dishonest, he holds $(z_0 \oplus x_0, z_1 \oplus x_1)$ and aims to guess $(z_0,z_1)$. This is equivalent to Bob guessing $(x_0, x_1)$ which he can do with probability $B_{OT}(P)$ Therefore $B_{OT}(Q)=B_{OT}(P)$.

	\subsection{ROT from Semi-random OT}
	
	Let $P$ be a Semi-random OT protocol with cheating probabilities $A_{OT}(P)$ and $B_{OT}(P)$. We construct a ROT protocol $Q$ with the same cheating probabilities as follows:
	\begin{enumerate}
		\item Alice picks $x_0 , x_1 \in \{0, 1\}$ uniformly at random.
		\item Alice and Bob perform the Semi-random OT protocol $P$ where Alice inputs $x_0, x_1$. Let $(c, y)$ be Bob's outputs.
		\item Alice and Bob abort in $Q$ if and only if they abort in $P$. Otherwise, the outputs of protocol $Q$ are $(x_0, x_1)$ for Alice and $(c, y)$ for Bob.
	\end{enumerate}
	
	The outputs of $Q$ are uniformly random bits (if both parties are honest) since Alice chooses her input at random. Note that, in the definition of ROT, the outputs are only required to be random in the honest case, and no assertions are made when one party acts dishonestly. Therefore $Q$ does indeed implement ROT. From the construction of $Q$ it is also clear that $A_{OT}(P)=A_{OT}(Q)$ and $B_{OT}(Q)=B_{OT}(P)$.
	
	\subsection{Semi-random OT from ROT in the general protocol framework}
	
	In order to fully motivate why the protocol framework in section \ref{sec:framework} is general for Semi-random OT, we here sketch how to recast Semi-random OT, realized by performing ROT together with the classical processing as detailed above in \ref{sec:SRfromROT}, in the form of our general framework. ROT with classical processing is not immediately in the form of the general protocol framework for Semi-random OT, since in a quantum protocol for ROT, Alice has outputs which she would obtain through a measurement. In the general protocol framework in \ref{sec:framework}, however, Alice makes no measurements. We also show that the cheating probabilities do not change when the protocol is recast.
	
	Suppose therefore that Alice obtains her two output bits in ROT by measuring a part of a quantum system held by her at some point during the protocol. (If desired, this measurement may be deferred to the end of the protocol, using the standard technique for this, closely related to the procedure we will describe below.)  Any POVM may be realized as a projective measurement in a suitably enlarged Hilbert space~\cite{neumark}, with as many dimensions as outcomes. We will label this Hilbert space $C$. Suppose therefore that in this possibly enlarged Hilbert space, Alice's four-outcome measurement has measurement operators $\Pi^{x_0,x_1}_C=|x_0,x_1\rangle_{CC}\langle x_0, x_1|$, which are orthonormal projectors on four orthogonal basis states $|x_0,x_1\rangle_C$ for $x_0, x_1 \in \{0, 1\}$. (The construction below can easily be extended to the case where Alice's four measurement operators are orthogonal projectors onto more than one basis state, that is, have rank $> 1$).
	
	Now, instead of measuring system $C$ to obtain $(x_0, x_1)$ and sending $(z_0 \oplus x_0, z_1 \oplus x_1)$ to Bob, where $(z_0, z_1)$ are Alice's inputs, Alice performs one of the four unitary transforms
	\begin{eqnarray}
	U^{z_0, z_1}_{CD} = \sum_{x_0, x_1 \in \{0,1\}} &&|x_0,x_1\rangle_{CC}\langle x_0, x_1|\\
	&&\otimes |z_0 \oplus x_0, z_1 \oplus x_1\rangle_{DD} \langle aux|\nonumber
	\end{eqnarray}
	on system $C$ and an auxiliary system $D$, where $|aux\rangle_D$ is a ``blank" state that could e.g. be chosen as $|0, 0\rangle$. The states $|0,0\rangle_D, |0,1\rangle_D, |1,0\rangle_D, |1,1\rangle_D$ form an orthonormal basis for the four-dimensional $D$ system. She then sends system $D$ to Bob, who (if he is honest) can measure this system to obtain $(z_0 \oplus x_0, z_1 \oplus x_1)$.
	
	This modified protocol for Semi-random OT is now in the form of the general framework. (If desired, Bob's measurements to obtain $(z_0 \oplus x_0, z_1 \oplus x_1)$ and $(c, y)$ can be combined into a single measurement by Bob that directly gives $(c, y')$.) By no-signalling~\cite{rimini, bussey, jordan, nosig1, nosig2}, Bob cannot tell whether or not Alice has measured system $C$. Therefore, Bob's cheating probability remains the same as if an honest Alice simply had measured system $C$ and sent him the state $|z_0 \oplus x_0, z_1 \oplus x_1\rangle$. Equivalently, Bob's cheating probability is the same as if Alice had measured system $C$ and sent him the classical bits $(z_0 \oplus x_0, z_1 \oplus x_1)$. Since the recast Semi-random OT protocol is otherwise the same as the ROT protocol we started with, in particular, how Bob obtains $(c,y)$ remains the same, Alice's cheating probabilities are also equal in both versions of the Semi-random protocol. That is, cheating probabilities remain the same in the version that is in the form of the general framework, and in the version where Alice and Bob perform ROT with classical processing.\\

\section{Bob's cheating probability for symmetric sets of states}
\label{app:symmBOT}
We need to obtain Bob's cheating probability for a symmetric set of four equiprobable pure states
$\sigma^{ij}_{BM}=|\psi^{ij}\rangle\langle\psi^{ij}|$, where ``symmetric'' means that there exists a unitary transform $U$ such that $U^4 = \mathbbm{1}$, and successive applications of $U$ to a ``starting state" will result in the other states in the set. It could either hold that $|\psi^{01}\rangle= U|\psi^{00}\rangle, |\psi\rangle^{11}=U^2|\psi^{00}\rangle,  |\psi^{10}\rangle= U^3|\psi^{00}\rangle$, which we will refer to as ``Case 1", or, that $|\psi^{11}\rangle= U|\psi^{00}\rangle, |\psi\rangle^{01}=U^2|\psi^{00}\rangle,  |\psi^{10}\rangle= U^3|\psi^{00}\rangle$, which we will refer to as ``Case 2". All other orderings will be equivalent to these two cases. Case 1 will result in a lower cheating probability for Bob for a given largest pairwise fidelity $F$ between two of the four states. That is, Case 1 will give 1-out-of-2 OT protocols with better performance.

In either case, Bob's optimal measurement is the minimum-error measurement for distinguishing between these four states. For a set of symmetric equiprobable states the optimal minimum-error measurement is the so-called square-root measurement. Its success probability for pure symmetric states can be obtained in terms of the sum of the square roots of the Gram matrix for the states~\cite{WDA14}. The elements of the Gram matrix for a set of states $\{|\psi_j\rangle\}$ are given by $G_{ij} = \langle{\psi_i}|{\psi_j}\rangle$.
For four symmetric pure states, the Gram matrix is given by
\begin{equation}
\label{eq:generalgram}
G=\left(\begin{array}{cccc}1 & f & G & f^* \\f^* & 1 & f & G \\G & f^* & 1 & f \\f & G & f^* & 1\end{array}\right),
\end{equation}
where $f$ is generally complex but $G$ is always real. In Case 1, it holds that $f=\langle\psi^{00}|\psi^{01}\rangle=\langle\psi^{01}|\psi^{11}\rangle = \langle\psi^{11}|\psi^{10}\rangle = \langle\psi^{10}|\psi^{00}\rangle$, and $G=\langle\psi^{00}|\psi^{11}\rangle = \langle\psi^{01}|\psi^{10}\rangle$. For sets of states that allow us to implement 1-out-of-2 oblivious transfer, it will in Case 1 also hold that $G=0$. As already mentioned, this follows from conditions \eqref{eq:c0} and \eqref{eq:c1}. In Case 1 it also then holds that the largest pairwise fidelity between two of the states $F=|f|$. In Case 2, it will instead hold that $f=0$ and $G$ is nonzero, with $|G|$ equal to the largest pairwise fidelity $F$.

The eigenvalues of the Gram matrix are equal to
\begin{eqnarray}
\lambda_0 &=& 1+f+G+f^*,~~ \lambda_1 = 1+if-G-if^*,\nonumber\\
 \lambda_2 &=& 1-f+G-f^*,~~ \lambda_3 = 1-if-G+if^*.
\end{eqnarray}
These eigenvalues are all real, and can also be shown to always be nonnegative.
The success probability for the square-root measurement, and hence Bob's cheating probability, is given by~\cite{WDA14}
\begin{eqnarray}
\label{eq:symbobcheat}
B_{OT} &=& \frac{1}{16}\left(\sqrt{\lambda_0}+\sqrt{\lambda_1}+\sqrt{\lambda_2}+\sqrt{\lambda_3} \right)^2\\
&=&\frac{1}{16}\left(\sqrt{1+G + 2 \text{Re} f}+\sqrt{1+G - 2 \text{Re} f} \right.\nonumber\\
&&+\left.\sqrt{1-G + 2 \text{Im} f}+\sqrt{1-G - 2 \text{Im} f} \right)^2. \nonumber
\end{eqnarray}
(Since the eigenvalues of the Gram matrix are nonnegative, the arguments of each of the square roots are non-negative.)

In Case 1, where $G=0$, Bob's optimal cheating probability becomes
\begin{eqnarray}
\label{eq:symbobcheatGzero}
B_{OT} &=&\frac{1}{16}\left(\sqrt{1 + 2 \text{Re} f}+\sqrt{1 - 2 \text{Re} f} \right.\nonumber\\
&&+\left.\sqrt{1+ 2 \text{Im} f}+\sqrt{1 - 2 \text{Im} f} \right)^2. \nonumber
\end{eqnarray}
Since Alice can always cheat at least with probability $A_{OT}\ge (1+F)/2$, the interesting range is $F=|f|\le 1/2$.
It is relatively easy to show that the expression in the RHS is then minimised when $f$ is pure real or pure imaginary.To show this, one can e.g. set $f=F\cos\theta + i F\sin\theta$, and differentiate with respect to $\theta$. For fixed $F\le 1/2$, $B_{OT}$ reaches its maximum value when $\theta = k\pi/2$ and its minima when $\theta=\pi/4+k\pi/2$, where $k$ is an integer.  (In the general case, where $|f|$ can be larger than $1/2$, we should make either $\text{Re} f$ or $\text{Im} f$ as large as possible, without the arguments of any of the square roots being negative, in order to minimise $B_{OT}$.) That is, in case 1 the smallest possible cheating probability for Bob when $F\le 1/2$ is equal to
\begin{eqnarray}
\label{eq:symmsmallestBOT}
B_{OT} =\frac{1}{4} \left(1+\frac{1}{2}\sqrt{1 + 2F}+\frac{1}{2}\sqrt{1 - 2F} \right)^2
\end{eqnarray}
as a function of the largest pairwise fidelity $F$ between the states.

In Case 2, Bob's optimal cheating probability is instead given by
\begin{eqnarray}
B_{OT} =\frac{1}{4} \left(\sqrt{1 + G}+\sqrt{1 - G} \right)^2,
\end{eqnarray}
where now the largest pairwise fidelity $F=|G|$. For a given $F$, this cheating probability for Bob is always larger than the one in Eq. \eqref{eq:symmsmallestBOT}.
To summarise, for a given largest pairwise fidelity $F\le 1/2$, Bob's cheating probability $B_{OT}$, for a set of equiprobable pure symmetric states, is at least as large as the cheating probability given in Eq. \eqref{eq:symmsmallestBOT}.

	\section{Alice's optimal cheating strategy in the example protocol}
	\label{app:alice}
	
	Alice, to pass a test by Bob with certainty, has to send a state of the form
	\begin{eqnarray}
	\label{eq:cheatstate}
	\ket{\psi_{\rm ch}} &=& a\ket{0}_A\otimes\ket{00}_B + b\ket{1}_A\otimes\ket{++}_B\nonumber\\
	&& + c\ket{2}_A\otimes\ket{11}_B + d\ket{3}_A\otimes\ket{- -}_B,
	\end{eqnarray}
	where $\{\ket 1_A, \ket 2_A, \ket 3_A, \ket 4_A\}$ is an orthonormal basis for a system $A$ she retains while sending Bob system $B$, and $|a|^2+|b|^2+|c|^2+|d|^2=1$.
	
	Bob measures the first $B$ qubit in the $Z$ basis and the second $B$ qubit in the $X$ basis. It holds that
	\begin{eqnarray}
	\braket{0+|\psi_{\rm ch}} &=& \frac{1}{\sqrt 2}\left(a\ket 0_A+b \ket 1_A\right)\nonumber\\
	\braket{1+|\psi_{\rm ch}} &=& \frac{1}{\sqrt 2}\left(b\ket 1_A+c \ket 2_A\right)\nonumber\\
	\braket{0-|\psi_{\rm ch}} &=& \frac{1}{\sqrt 2}\left(a\ket 0_A+d \ket 3_A\right)\nonumber\\
	\braket{1-|\psi_{\rm ch}} &=& \frac{-1}{\sqrt 2}\left(c\ket 2_A+d \ket 3_A\right).
	\label{eq:condstates}
	\end{eqnarray}
	These states are the unnormalised states conditionally prepared on Alice's side, given Bob's measurement outcome. The norm of each of the above states gives the probability for that outcome on Bob's side. That is, it is the probability with which the corresponding state is prepared.
	
	To successfully cheat, Alice needs to determine whether Bob received the first or second bit. Bob obtains the first bit if he obtains $(0,+)$ or $(1,-)$, and the second bit if he obtains $(0,-)$ or $(1,+)$. It so happens that each of these outcome combinations occur with probability 1/2, irrespective of $a, b, c, d$. The two density matrices Alice needs to distinguish between are $\rho_0$ and $\rho_1$, with
	\begin{eqnarray}
	\frac{1}{2}\rho_0 = \braket{0+|\psi_{\rm ch}} \braket{\psi_{\rm ch}|0+}
	+\braket{1-|\psi_{\rm ch}} \braket{\psi_{\rm ch}|1-},\nonumber\\
	\frac{1}{2}\rho_1 = \braket{0-|\psi_{\rm ch}} \braket{\psi_{\rm ch}|0-}
	+\braket{1+|\psi_{\rm ch}} \braket{\psi_{\rm ch}|1+},\nonumber\\
	\end{eqnarray}
	which in matrix form, with the basis states ordered $\{\ket 0_A, \ket 1_A, \ket 2_A, \ket 3_A\}$, are given by
	\begin{eqnarray}
	\rho_0 &=& \left(\begin{array}{cccc}|a|^2 & ab^* & 0 & 0 \\
	a^*b & |b|^2 & 0 & 0 \\
	0 & 0 & |c|^2 & cd^* \\
	0 & 0 & c^*d & |d|^2\end{array}\right),\nonumber\\
	\rho_1 &=& \left(\begin{array}{cccc}|a|^2 & 0 & 0 & ad^* \\
	0 & |b|^2 & bc^* & 0 \\
	0 & b^*c & |c|^2 & 0 \\
	a^*d & 0 & 0 & |d|^2\end{array}\right).
	\end{eqnarray}
	Alice's optimal measurement is the Helstrom measurement, given by a projection in the eigenbasis of $\rho_0-\rho_1$. If Alice obtains an outcome corresponding to a positive eigenvalue, she guesses that Bob obtained the first bit, and if she obtains an outcome corresponding to a negative eigenvalue, then she guesses that Bob obtained the second bit. If Alice obtains an outcome corresponding to a zero eigenvalue, she can guess either the first or second bit, without altering her success probability (conditioned on such an outcome, Bob is equally likely to have obtained the first or second bit). Because the state space on Bob's side is three-dimensional, the situation is effectively three-dimensional on Alice's side too, but it is convenient to keep $\{\ket 0_A, \ket 1_A, \ket 2_A, \ket 3_A\}$ as a basis.
	
	We therefore need to find the eigenvalues of
	\begin{equation}
	\rho_0-\rho_1 = \left(\begin{array}{cccc}0 & ab^* & 0 & -ad^* \\
	a^*b & 0 & -bc^* & 0 \\
	0 & -b^*c & 0 & cd^* \\
	-a^*d & 0 & c^*d & 0\end{array}\right).
	\end{equation}
	The eigenvalues are 
	\begin{eqnarray}
	\label{eq:eigenvalues}
	\lambda_1 = \lambda_2 =0, \lambda_{3,4} &=&  \pm \sqrt{|ab|^2+|bc|^2+|cd|^2+|ad|^2}\nonumber\\
	&=&\pm\sqrt{(|a|^2+|c|^2)(|b|^2+|d|^2)}.
	\end{eqnarray}
	where we choose the $+$ sign for $\lambda_3$. The success probability is therefore given by
	\begin{eqnarray}
	p_{\rm cheat} &=& \frac{1}{2}+\frac{1}{4} {\rm Tr}  |\rho_0-\rho_1| = \frac{1}{2}+\frac{1}{4} \sum_i|\lambda_i|\nonumber\\
	&=& \frac{1}{2}[1+\sqrt{(|a|^2+|c|^2)(|b|^2+|d|^2)}].
	\end{eqnarray}
	Clearly, Alice's cheating probability is maximised when $|a|^2+|c|^2 = |b|^2+|d|^2 = 1/2$, giving a maximum cheating probability of 3/4 whenever this condition is met. One optimal choice for Alice is for example $|a|=|b|=1/\sqrt 2$ and $c=d=0$. In this case, $\rho_0 = \ket + \bra{+}$ and $\rho_1 = 1/2(\ket 0 \bra{0} + \ket 1\bra{1})$. Alice should measure in the $\ket +, \ket -$ basis, where $\ket \pm = (\ket 0 + \ket 1)/\sqrt 2$. With probability $1/4$, she will obtain the outcome ``$-$", and is then sure that Bob obtained the second bit (outcomes $(0,-)$ or $(1,+)$ for Bob). With probability $3/4$, she will obtain the outcome ``$+$", and then she guesses that Bob obtained the first bit. Her guess is in this case however only correct with probability 2/3, giving an overall cheating probability of 3/4.
	
	Choosing either $|a|$ or $|c|$ equal to $1/\sqrt 2$ and the other one equal to zero, and either $|b|$ or $|d|$ equal to $1/\sqrt 2$ and the other one equal to zero gives the same cheating probability. These optimal cheating strategies all require only a two-dimensional system on Alice's side.
	Choosing $|a|= |b|= |c|= |d|=1/2$ also gives $p_{\rm cheat} = 3/4$; these are examples of cheating states with high symmetry.
	As an example of a suboptimal cheating strategy, choosing three of the parameters equal to $1\sqrt 3$ and the remaining one equal to zero gives $p_{\rm cheat} = 1/2(1+ \sqrt 2/3)$, which is less than 3/4.

	\section{Preparation of Alice's entangled state}
	\label{app:cheatstateprep}
	In this appendix we will in detail describe the preparation of the state (\ref{eq:optcheatstate}), 
	$$
	|\Sigma\rangle = \frac{1}{\sqrt{2}} \left( |00\rangle_B|0\rangle_A +|++\rangle_B|1\rangle_A\right).
	$$
	This state can be prepared by means of a controlled-phase gate $U_{\rm{CP}}$, a Hadamard gate $H$, a controlled-controlled-phase gate $U_{\rm{CCP}}$, and local unitary operations.
	Controlled phase gates introduce tunable and conditional phase shifts. Specifically, 
	\begin{align*} 
	U_{\rm{CP}}  &= \mathbb{I} + [\exp(i\alpha) - 1]|11\rangle\langle11|, \\ 
	U_{\rm{CCP}} &= U_{\rm{CCP}} = \mathbb{I} + [\exp(i\beta) - 1]|111\rangle\langle111|.
	\end{align*}
	We used the quantum circuit in Fig.~\ref{fig:1}a to turn an initially separable state $|\psi_{\rm{in}}\rangle$ into a state which is equivalent to $|\Sigma\rangle$ up to local unitary operations. The parameters $\alpha, \beta$ describe the net operation $U(\alpha, \beta) = U_{\rm{CCP}}(\beta)(\mathbb{I} \otimes H \otimes \mathbb{I}) (U_{\rm{CP}}(\alpha)\otimes \mathbb{I})$.
	
	The input state can be parametrized by two tuples of angles, $\boldsymbol{\theta} = \{\theta_{i=1, 2, 3}\}$ and $\boldsymbol{\phi} = \{\phi_{i=1, 2, 3}\}$, as
	\begin{equation*}
	\label{eq:inputparam}
	|\psi_{\rm{in}}\rangle = \prod\limits_{i=1,2,3}^{\otimes} \left[\cos(\theta_i/2)|0\rangle + \sin(\theta_i/2)e^{i\phi_i}|1\rangle\right].
	\end{equation*}
	The degree of local-unitary equivalence $E(|\rm{a}\rangle, |\rm{b}\rangle)$ between states $|\rm{a}\rangle$ and $|\rm{b}\rangle$ can be quantified by an overlap maximized over all local unitary operations
	\begin{equation*}
	\label{eq:LUE}
	E(|\rm{a}\rangle, |\rm{b}\rangle) = \max_{\bf{v}} |\langle a | V_{\rm{LO}}({\bf{v}}) | b \rangle|^2,
	\end{equation*}
	where ${\bf{v}}$ is a tuple containing 9 parameters $\{A_j, B_j, C_j\}_{j=1,2,3}$ which parametrize the operation $V_{\rm{LO}} = V_{1}\otimes V_2 \otimes V_3$. Specifically, the parameters $A_j, B_j,$ and $C_j$ describe a $j$-th local operation
	\begin{equation}
	\label{eq:unitary}
	V_{j} = \left(\begin{array}{ll}
	\cos(A_j)\exp(i B_j) & -\sin(A_j)\exp(-iC_j) \\
	\sin(A_j)\exp(iC_j) & \cos(A_j)\exp(-i B_j)
	\end{array}\right).
	\end{equation}
	We maximized $E(|\Sigma\rangle, U(\alpha, \beta)|\psi_{\rm{in}}(\boldsymbol{\theta}, \boldsymbol{\phi})\rangle)$ numerically using the Broyden-Fletcher-Goldfarb-Shanno (BFGS) algorithm \cite{BFGS}.
	
	First we performed the optimization with all parameters being free and with multiple random initial guesses. From the set of optima we arbitrarily picked the parameter-tuples with $\theta_1 \approx 120^{\circ}$, fixed $\theta_1 = 120^{\circ}$ and performed the optimization again. We repeated this procedure to gradually fix also $\phi_1$, $\theta_2$, $\phi_2$, $\beta$, and $\varphi_3$, in this order. The parameters $\alpha$ and $\theta_3$ remained free in the last round of the optimization. The optimal parameters are listed in Tab.~\ref{tab:parameters}. With these parameters, the complement of $E$ to one is sufficiently small, $1 - E \approx 8\cdot10^{-11}$. 
	\begin{table}[h]
		\begin{tabular}{l|r||l|r}
			$\theta_1$	& $120.000^{\circ}$ 			& $\phi_1$ 	& $22.500^{\circ}$ \\
			$\theta_2$	& $90.000^{\circ}$ 			& $\phi_2$ 	& $90.000^{\circ}$ \\
			$\theta_3$	& $116.565^{\circ}$ 		& $\phi_3$ 	& $180.000^{\circ}$ \\
			$\alpha$	& $-138.190^{\circ}$	& $\beta$ 	& $180.000^{\circ}$
		\end{tabular}
		\caption{Optimal parameters for preparation of state $|\Sigma\rangle$ .}
		\label{tab:parameters}
	\end{table}
	
	Next, we initialized the circuit and the input state with the optimal parameters and performed tomography of the output quantum state. Employing the maximum-likelihood method \cite{Jezek2003} we reconstructed the density matrix $\rho_{\rm{exp},0}$ of actually prepared quantum state. Then we numerically maximized the expectation value
	\begin{equation*}
	\label{eq:correctionoverlap}
	\langle\Sigma|U_{\rm{LO}}(\mathbf{u})\rho_{\rm{exp},0} U_{\rm{LO}}^{\dagger}(\mathbf{u})|\Sigma\rangle
	\end{equation*}
	to find the corrective local operations $U_{\rm{LO}}$. The optimal $U_{\rm{LO}}$ not only implements the required local operation to finish the preparation of $|\Sigma\rangle$, but also compensates for some systematic errors. The parameters of the optimal unitaries are listed in Tab.~\ref{tab:Uparam}. We parametrize $U_{\rm{LO}} = U_{1} \otimes U_{2} \otimes U_{3}$ the same way as in case of $V$, see Eq.~(\ref{eq:unitary}). Note that these parameters are not unique, multiple solutions exist (due to insensitivity to global phase and phase periodicity).
	
	\begin{table}[h]
		\begin{tabular}{l|r|r|r}
			$i$ & $A_i$ [deg]     & $B_i$ [deg]      & $C_i$ [deg]      \\
			\hline
			1 & 41.315 & 49.770   & 136.535  \\
			2 & 48.385 & -37.718 & 42.637   \\
			3 & 29.367 & -1.225  & -177.329
		\end{tabular}
		\caption{Parameters of the corrective unitary operations.}
		\label{tab:Uparam}
	\end{table}
	
	An arbitrary unitary operation acting on a single polarization qubit can be easily implemented by a sequence of a quarter-wave plate, half-wave plate, and another quarter-wave plate. However, we merged the unitary $U_{\rm{LO}}$ into final projective measurements. It can be done because the output state is projected at the end onto a state $|\pi\rangle$ and the projection $\langle\pi|U_{i}|\eta\rangle$ is equivalent to $\langle\tilde{\pi}|\eta\rangle$ with $|\tilde{\pi}\rangle = U_{i}^{\dagger}|\pi\rangle$. We found the corresponding wave-plate angles for six-state tomography by means of numerical minimization. They are listed in Tab.~\ref{tab:WPtransformedAngles}. 
	This optimization reduces the number of components in the experimental setup, reducing experimental imperfections and losses which accumulate with each added component. 
	\begin{table}[h]
		\begin{tabular}{l|rr|rr|rr}
			$|\pi\rangle$ & HWP1 & QWP1  & HWP2 & QWP2  & HWP3 & QWP3   \\
			\hline
			$|0\rangle$  & -15.35 & 49.83  & 9.80   & 53.28  & 80.78 & 8.54   \\
			$|1\rangle$  & 29.65  & -40.17 & 91.52 & -53.28 & 27.24 & -8.54  \\
			$|+\rangle$  & 3.16   & 94.65  & 47.90  & 92.29  & 9.80   & 92.26  \\
			$|-\rangle$  & 43.50   & 85.35  & 2.90   & 2.29   & -35.20 & 2.26   \\
			$|R\rangle$  & 22.80   & -1.28  & 64.96 & 82.06  & 9.51  & 53.85  \\
			$|L\rangle$  & -20.93 & 1.28   & 19.96 & -7.94  & 54.51 & -36.15
		\end{tabular}
		\caption{Wave-plate angles for transformed projectors. All numbers are in degrees.}
		\label{tab:WPtransformedAngles}
	\end{table}

	\section{Experimental data}
	\label{app:expdata}
	In this appendix we present the full sets of experimental data. The tables contain measured counts $C$, relative frequencies (or estimated probabilities) $f$, and theoretically predicted probabilities $p_t$. Relative frequencies were calculated as a ratio of the number of respective counts to the total number of counts. Digits in parenthesis represent one standard deviation at the final decimal place.
	The statistical errors were computed using error-propagation and the fact that the count-rates obey Poisson distribution. 
	
	Table~\ref{tab:CorrectTransfer} shows data for the case when both parties were honest. Alice sent states $|00\rangle$, $|++\rangle$, $|--\rangle$, $|11\rangle$ and Bob measured in the $ZX$ basis. Table~\ref{tab:HonestAlarms} shows data for Bob's test measurements when he measured the incoming states in the $XX$ or $ZZ$ basis. 
	
	Table~\ref{tab:BobCheating} summarizes results for the situation when Alice was honest but Bob was cheating. This means that Bob has been performing square-root measurements. 
	
	The situation when Bob was honest but Alice was cheating is recorded in the last two tables. Table~\ref{tab:AliceCheating1} shows the relative frequencies of Alice's  correct and incorrect estimates of the values of Bob's bit $c$. 
	Table~\ref{tab:AliceCheating2} shows relative frequencies of different results of Alice's and Bob's measurements in the test phase of the protocol.
	Theoretically, Bob should only detect $|++\rangle$ or $|00\rangle$.

	\begin{table}[h]
		\begin{tabular}{l|l|rrrr|r}
			& $|\psi_{B}\rangle$ & \multicolumn{4}{c|}{$|\pi_{B}\rangle$} & $p_{\rm{s}}$\\
			\hline
			&                    & $|0+\rangle$       & $|0-\rangle$        & $|1+\rangle$       & $|1-\rangle$       &         \\
			\hline                    
			$C$ & $|00\rangle$ & 892      & 829       & 3        & 3        &         \\
			$f$   &                    & 0.52(1)  & 0.48(1)   & 0.002(1) & 0.002(1) & 1.00(2) \\
			$p_t$ &                    & 0.5    & 0.5     & 0    & 0    & 1       \\
			\hline
			$C$ & $|++\rangle$ & 823      & 2         & 782      & 7        &         \\
			$f$ &                    & 0.51(1)  & 0.0012(9) & 0.48(1)  & 0.004(2) & 0.99(2) \\
			$p_t$ &                    & 0.5    & 0     & 0.5    & 0    &  1       \\
			\hline                    
			$C$ & $|--\rangle$  & 7        & 824       & 15       & 867      &         \\
			$f$ &                    & 0.004(2) & 0.48(1)   & 0.009(2) & 0.51(1)  & 0.99(2) \\
			$p_t$ &                    & 0    & 0.5     & 0   & 0.5    &   1      \\
			\hline                    
			$C$ & $|11\rangle$ & 0        & 1         & 800      & 841      &         \\
			$f$ &                    & 0.000(0) & 0.0006(5) & 0.49(1)  & 0.51(1)  & 1.00(2) \\
			$p_t$ &                    & 0    & 0     & 0.5    & 0.5    &     1   
		\end{tabular}
		\caption{Measured counts $C$, relative frequencies $f$, and corresponding theoretical probabilities $p_t$ for the situation when both the parties were honest. $|\psi_{B}\rangle$ is a state which Alice sends to Bob. Bob measures projection onto $|\pi_{B}\rangle$. Here, $p_s$ is the probability of correct receipt, i.e. Bob gets erroneous bit with probability $1-p_s$.}
		\label{tab:CorrectTransfer}
	\end{table}

	\begin{table}[h]
		\begin{tabular}{l|l|rrrr|r}
			& $|\psi_{B}\rangle$ & \multicolumn{4}{c|}{$|\pi_{B}\rangle$} & $p_{\rm{FA}}$\\
			\hline
			& & $|00\rangle$       & $|01\rangle$        & $|10\rangle$        & $|11\rangle$        &          \\
			\hline
			$C$ & $|00\rangle$ & 1701     & 3         & 1         & 0         &          \\
			$f$ &                     & 0.998(1) & 0.002(1)  & 0.0006(5) & 0.000(0)  & 0.002(1) \\
			$p_t$ &                    & 1    & 0     & 0     & 0     &        0  \\
			\hline                    
			$C$ & $|11\rangle$ & 0        & 0         & 15        & 1592      &          \\
			$f$ &                    & 0.000(0) & 0.000(0)  & 0.009(2)  & 0.991(2)  & 0.009(2) \\
			$p_t$ &                    & 0    & 0     & 0     & 1     &        0  \\
			\hline
			& & $|++\rangle$       & $|+-\rangle$        & $|-+\rangle$        & $|--\rangle$        &          \\
			\hline
			$C$ & $|++\rangle$ & 1615     & 1         & 43        & 1         &          \\
			$f$ &                      & 0.973(4) & 0.0006(5) & 0.026(4)  & 0.0006(5) & 0.027(4) \\
			$p_t$ &                      & 1    & 0     & 0     & 0     &        0  \\
			\hline                      
			$C$ & $|--\rangle$   & 5        & 9         & 9         & 1660      &          \\
			$f$ &                     & 0.003(1) & 0.005(2)  & 0.005(2)  & 0.986(3)  & 0.014(3) \\
			$p_t$ &                     & 0    & 0     & 0     & 1     &         0
		\end{tabular}
		\caption{Data for Bob's test measurements in the case when Alice was honest. Here, $p_{\rm{FA}}$ is the probability of ``false alarm'', i.e. the probability that Bob aborts the protocol even if Alice is not cheating.}
		\label{tab:HonestAlarms}
	\end{table}

	\begin{table}[h]
		\begin{tabular}{l|l|rrrr|l}
			& $|\psi_{B}\rangle$ & \multicolumn{4}{c|}{$|\pi_{B}\rangle$}           &  $p_{\rm{CE}}$        \\
			\hline
			&             & $|\zeta_{0}\rangle|\xi_{0}\rangle$    & $|\zeta_{0}\rangle|\xi_{1}\rangle$    & $|\zeta_{1}\rangle|\xi_{0}\rangle$    & $|\zeta_{1}\rangle|\xi_{1}\rangle$    &          \\
			\hline 
			$C$ & $|00\rangle$      & 85       & 1215     & 4        & 114      &          \\
			$f$                    &  & 0.060(6) & 0.857(9) & 0.003(1) & 0.080(7) & 0.857(9) \\
			$p_t$                    &   & 0.125    & 0.729    & 0.021    & 0.125    &    0.729      \\
			\hline 
			$C$ & $|++\rangle$      & 1013     & 184      & 301      & 53       &         \\
			$f$ &  & 0.65(1)  & 0.119(8) & 0.19(1)  & 0.034(5) & 0.65(1)  \\
			$p_t$ &      & 0.729    & 0.125    & 0.125    & 0.021    &      0.729    \\
			\hline 
			$C$  & $|--\rangle$      & 64       & 253      & 384      & 1441     &          \\
			$f$ &  & 0.030(4) & 0.118(7) & 0.179(8) & 0.67(1)  & 0.67(1)  \\
			$p_t$ &  & 0.021    & 0.125    & 0.125    & 0.729    &      0.729    \\
			\hline 
			$C$ & $|11\rangle$      & 228      & 48       & 1360     & 253      &          \\
			$f$ &  & 0.121(7) & 0.025(4) & 0.72(1)  & 0.134(8) & 0.72(1)  \\
			$p_t$ &    & 0.125    & 0.021    & 0.729    & 0.125    &    0.729     
		\end{tabular}
		\caption{Bob was cheating, Alice was honest. Here, $p_{\rm{CE}}$ is the probability of Bob correctly estimating the incoming state.}
		\label{tab:BobCheating}
	\end{table}

	\begin{table}[h]
		\begin{tabular}{ll|lll}
			Alice's & Bob's &   &   &   \\
			estimate & bit $c$ & $C$ & $f$ & $p_t$ \\
			\hline
			0            & 0          & 856    & 0.53(1)     & 0.5  \\
			0            & 1          & 356    & 0.22(1)     & 0.25  \\
			1            & 0          & 17     & 0.010(3)    & 0  \\
			1            & 1          & 400    & 0.25(1)     & 0.25 
		\end{tabular}
		\caption{Alice was cheating, Bob was honest. The table shows the probabilities of Alice correctly/incorrectly guessing Bob's bit $c$.}
		\label{tab:AliceCheating1}
	\end{table}

	\begin{table}[h]
		\begin{tabular}{l|l|lll}
			$|\pi_{A}\rangle$ &  $|\pi_{B}\rangle$ & $C$ & $f$ & $p_t$ \\
			\hline
			$|0\rangle$ & $|00\rangle$        & 851   & 0.52(1)      & 0.5  \\
			$|0\rangle$ & $|01\rangle$        & 15     & 0.009(2)    & 0 \\
			$|0\rangle$ & $|10\rangle$        & 28     & 0.017(3)    & 0  \\
			$|0\rangle$ & $|11\rangle$        & 31     & 0.019(3)    & 0  \\
			$|1\rangle$ & $|++\rangle$        & 688   & 0.42(1)      & 0.5  \\
			$|1\rangle$ & $|+-\rangle$        & 7       & 0.004(2)    & 0  \\
			$|1\rangle$ & $|-+\rangle$        & 11     & 0.007(2)    & 0  \\
			$|1\rangle$ & $|--\rangle$        & 4       & 0.002(1)    & 0 
		\end{tabular}
		\caption{Test measurements for an honest Bob when Alice was cheating. Alice measured her qubit in the $Z$ basis and Bob measured his qubit in the $ZZ$ or $XX$ basis.}
		\label{tab:AliceCheating2}
	\end{table}
	
\end{document}